\pdfoutput=1
\documentclass[11pt]{article}

\let\originalleft\left
\let\originalright\right
\renewcommand{\left}{\mathopen{}\mathclose\bgroup\originalleft}
\renewcommand{\right}{\aftergroup\egroup\originalright}

\usepackage{algorithm,algpseudocode}

\usepackage{amsmath, amssymb, amsthm}
\usepackage{enumerate,enumitem} \setlist{leftmargin=*, itemsep=0pt}
\usepackage[figurewithin=none]{caption}
\usepackage{autobreak, bbm, booktabs, comment, dsfont, float, geometry, graphicx, mathtools, mathrsfs,  soul, suffix, tabu, tikz, tikzsymbols, titling, thmtools, thm-restate, tocloft, verbatim, xcolor, xspace}

\usepackage[backend=biber, isbn=false, backref, maxbibnames=99]{biblatex}
\DefineBibliographyStrings{english}{
	backrefpage = {p\adddot},
	backrefpages = {pp\adddot}
}

\usepackage[colorlinks, allcolors=blue, linktocpage, breaklinks]{hyperref}
\usepackage[capitalize]{cleveref}
\crefname{ALG@line}{Line}{Lines}

\numberwithin{equation}{section}

\allowdisplaybreaks

\theoremstyle{plain}
\newtheorem{thm}{Theorem}[section]
\newtheorem{lem}[thm]{Lemma}
\newtheorem{cor}[thm]{Corollary}

\newtheorem{obs}[thm]{Observation}

\crefname{lem}{Lemma}{Lemmas}
\crefname{thm}{Theorem}{Theorems}
\crefname{cor}{Corollary}{Corollaries}
\crefname{clm}{Claim}{Claims}
\crefname{prp}{Proposition}{Propositions}
\crefname{xmp}{Example}{Examples}
\crefname{que}{Question}{Questions}

\theoremstyle{definition}
\newtheorem{dfn}[thm]{Definition}

\newtheorem{prb}[thm]{Problem}
\crefname{prb}{Problem}{Problems}

\theoremstyle{remark}
\newtheorem*{rmk*}{Remark}

\newcommand*{\C}{\ensuremath{\mathbb{C}}}
\newcommand*{\N}{\ensuremath{\mathbb{N}}}

\newcommand*{\R}{\ensuremath{\mathbb{R}}}

\newcommand{\microspace}{\mspace{.5mu}}
\newcommand{\ket}[1]{\ensuremath{\left\lvert\microspace #1
		\microspace\right\rangle}}
\newcommand{\bra}[1]{\ensuremath{\left\langle\microspace #1
		\microspace\right\vert}}
\newcommand{\ketbra}[2]{\ensuremath{\left\lvert\microspace #1
		\microspace\right\rangle\! \left\langle \microspace #2 \microspace \right\rvert}}
\newcommand*{\kb}[1]{\ketbra{#1}{#1}}
\newcommand*{\ip}[2]{\left\langle #1 \middle| #2 \right\rangle}

\newcommand*{\acz}{$\mathsf{AC^0}$\xspace}
\newcommand*{\qacz}{$\mathsf{QAC^0}$\xspace}
\newcommand*{\qaczf}{$\mathsf{QAC_f^0}$\xspace}

\newcommand*{\qnco}{$\mathsf{QNC^1}$\xspace}

\newcommand*{\tcz}{$\mathsf{TC^0}$\xspace}

\usetikzlibrary{arrows, calc, positioning, shapes}

\newcommand*{\adj}[1]{#1^\dagger}
\newcommand*{\bits}{\{0,1\}}
\newcommand*{\cc}[1]{\ensuremath{\mathsf{#1}}}
\newcommand*{\class}[1]{\mathsf{#1}}

\newcommand*\cube[1]{\bits^{#1}}
\newcommand*{\eps}{\varepsilon}

\newcommand*{\interact}{\mathord{\rightleftarrows}}

\newcommand*{\mrm}[1]{\mathrm{#1}}
\newcommand*{\mc}{\mathcal}
\newcommand*{\norm}[1]{\|#1\|}
\newcommand*{\Norm}[1]{\left\|#1\right\|}

\newcommand*{\poly}{\mathrm{poly}}

\newcommand*{\pr}[1]{\mathrm{Pr}(#1)}

\newcommand*{\PR}[1]{\mathrm{Pr}\left(#1\right)}

\newcommand*{\reg}{\mathsf}

\newcommand*\tr[1]{\operatorname{tr} \left( #1 \right)}
\newcommand*\trg[2]{\operatorname{tr}_{> #1} \left( #2 \right)}
\DeclareMathOperator*{\td}{td}
\newcommand*\wt\widetilde
\newcommand*{\zs}{0\dotsc0}
\DeclarePairedDelimiter{\ceil}{\lceil}{\rceil}

\newcommand{\Paren}[1]{\left(#1\right)}

\newcommand*{\Mag}[1]{\left| #1 \right|}

\newcommand{\IP}{\class{IP}}

\newcommand{\PSPACE}{\class{PSPACE}}
\newcommand{\QIP}{\class{QIP}}

\newcommand{\statePSPACE}{\cc{statePSPACE}\xspace}
\newcommand{\sPe}{\cc{statePSPACE_{exp}}\xspace}

\makeatletter
\newcommand{\hlabel}[1]{\edef\@currentcounter{ALG@line}\label{#1}}
\makeatother

\algblockdefx{Ctrl}{EndCtrl}[1]{\textbf{controlled on} #1}{\textbf{end control}}

\floatname{algorithm}{Procedure}
\crefname{algorithm}{Procedure}{Procedures}
\algnewcommand{\IIf}[1]{\State\algorithmicif\ #1\ \algorithmicthen}
\algnewcommand{\EndIIf}{\unskip\ \algorithmicend\ \algorithmicif}

\newcommand*{\ft}{\mathbb F_2}
\DeclareMathOperator*{\im}{im}
\DeclareMathOperator*{\sgn}{sgn}
\usepackage{multirow,makecell}
\newcommand*{\mr}[1]{\multirow2*{#1}}
\usepackage{tablefootnote}

\newcommand{\stateqip}{\class{stateQIP}}
\newcommand{\linear}{\mathrm{D}}
\newcommand\ptr[2]{\mathrm{tr}_{#1} \Paren{#2}}

\newcommand{\Tr}{\mathrm{tr}}

\DeclareMathOperator{\sr}{sgnRe}
\newcommand*{\crl}[1]{\text{\textnormal{ctrl-}}#1}

\title{Efficient Quantum State Synthesis with One Query}
\author{Gregory Rosenthal\thanks{Email: \href{mailto:grosenth@uwaterloo.ca}{\color{black}\texttt{grosenth@uwaterloo.ca}}. Part of this work was done while the author was visiting the Simons Institute for the Theory of Computing.} \\ University of Waterloo}
\predate{} \postdate{} \date{}
\begin{document}
\maketitle
\begin{abstract}
We present a polynomial-time quantum algorithm making a single query (in superposition) to a classical oracle, such that for every state $\ket\psi$ there exists a choice of oracle that makes the algorithm construct an exponentially close approximation of $\ket\psi$. Previous algorithms for this problem either used a linear number of queries and polynomial time, or a constant number of queries and polynomially many ancillae but no nontrivial bound on the runtime. As corollaries we do the following:
\begin{itemize}
	\item We simplify the proof that $\cc{statePSPACE} \subseteq \cc{stateQIP}$ (a quantum state analogue of $\cc{PSPACE} \subseteq \cc{IP}$) and show that a constant number of rounds of interaction suffices.
	\item We show that \qaczf lower bounds for constructing explicit states would imply breakthrough circuit lower bounds for computing explicit Boolean functions.
	\item We prove that every $n$-qubit state can be constructed to within 0.01 error by an $O(2^n/n)$-size circuit over an appropriate finite gate set. More generally we give a size-error tradeoff which, by a counting argument, is optimal for \emph{any} finite gate set.
\end{itemize}
\end{abstract}
\newpage\tableofcontents\newpage
\section{Introduction} \label{sec:intro}
Many natural tasks in quantum computing can be phrased as \emph{constructing a quantum state} $\ket\psi$, by which we mean implementing a quantum circuit that outputs $\ket\psi$ given the all-zeros input. Examples include ground states of physical systems~\cite{cerezo2021variational}, Hamiltonian simulation applied to the all-zeros state~\cite{swingle2018unscrambling}, QSampling states~\cite{AT03}, quantum money~\cite{aaronson2009quantum}, quantum pseudorandom states~\cite{ji2018pseudorandom}, and the first step in Linear Combinations of Unitaries (LCU)~\cite{BCK15,CW12}. A survey by Aaronson~\cite{Aar16} discusses other examples.

Despite this, much less is known about the complexity of constructing quantum states than is known about the (quantum) complexity of computing Boolean functions. This motivates the goal of finding, for a state $\ket\psi$ that we would like to construct, a Boolean function $f$ such that the task of constructing $\ket\psi$ efficiently reduces to that of computing $f$. We can phrase this problem as follows:

\begin{prb}[The state synthesis problem~\cite{Aar16}, stated informally] \label{que:ssp}
	Find a low-complexity quantum algorithm $A$, which can make (adaptive) queries in superposition to a classical oracle, such that for every state $\ket\psi$ there exists an oracle $f$ such that $A^f$ approximately constructs $\ket\psi$.
\end{prb}

We call $A$ a \emph{state synthesis algorithm}. By $A^f$ we mean $A$ with query access to the Boolean function $f$. We may assume without loss of generality that $f$ has a single output bit, for reasons that will be explained in \cref{sec:prelim} when we define the query model. The requirement that all queries be to the \emph{same} function $f$ is also without loss of generality, because if the $j$'th query is to a function $f_j$ then the function $(j,x) \mapsto f_j(x)$ can simulate all queries.

Our main result is the following state synthesis algorithm, where by a ``clean construction" we mean that the ancillae end in approximately the all-zeros state (as opposed to some other state unentangled with $\ket\psi$):

\begin{thm}[Main theorem, informal] \label{thm:mi}
	There is a uniform sequence $(C_n)_n$ of $\poly(n)$-size quantum circuits, each making one (resp.\ four) queries to a classical oracle, such that for every $n$-qubit state $\ket\psi$ there exists a classical oracle $f$ such that $C_n^f$ non-cleanly (resp.\ cleanly) constructs $\ket\psi$ to within exponentially small error.
\end{thm}

We state \cref{thm:mi} formally in \cref{sec:msa}. The rest of the Introduction is organized as follows: in \cref{sec:ssa} we compare \cref{thm:mi} to previously known state synthesis algorithms, in \cref{sec:bar,sec:iss,sec:cub} we discuss three different applications of \cref{thm:mi}, and in \cref{sec:org} we discuss the organization of the rest of the paper.

First we briefly sketch the proof of \cref{thm:mi}.\footnote{We credit Fermi Ma~\cite{Ma23} for suggesting a simplification of the proof which he has allowed us to incorporate, as discussed in \cref{sec:post}.} For simplicity, in this sketch we allow the circuit (as opposed to just the oracle) to depend on the state $\ket\psi$ being constructed. Call a state of the form $C \cdot 2^{-n/2} \sum_{x \in \cube n} \pm \ket x$ where $C$ is a Clifford unitary a ``Clifford times phase state". Irani, Natarajan, Nirkhe, Rao and Yuen~\cite{INN+22} proved that every state has fidelity $\Omega(1)$ with some Clifford times phase state (more generally, this holds for $C$ from any 2-design) and observed that Clifford times phase states can be efficiently constructed with one query. Thus all that remains is to decrease the approximation error.

We recursively define $\ket{\phi_k}$ for $k \ge 0$ as a Clifford times phase state that has fidelity $\Omega(1)$ with $\Paren{\ket\psi - \sum_{j=0}^{k-1} c_j \ket{\phi_j}} / \Norm{\ket\psi - \sum_{j=0}^{k-1} c_j \ket{\phi_j}}$, for appropriately chosen coefficients $c_0, c_1, \dotsc$ tending to zero. We show that $\sum_{j=0}^{k-1} c_j \ket{\phi_j}$ is a good approximation of $\ket\psi$ for sufficiently large $k$. Furthermore, using Linear Combinations of Unitaries (LCU)~\cite{BCK15,CW12} we can construct this approximation of $\ket\psi$ with constant success probability. Finally we increase the success probability either by parallel repetition (in the one-query version of the theorem), with parallel queries merged into a single query, or by a hybrid of parallel repetition and amplitude amplification (in the four-query version).
\subsection{Comparison to previous state synthesis algorithms} \label{sec:ssa}

There is a trivial state synthesis algorithm using one query, where that query returns the description of a circuit over a universal gate set that constructs an exponentially close approximation of the target state $\ket\psi$. This construction can be made clean by using a second query to uncompute the first query after constructing $\ket\psi$. However this algorithm requires an exponential number of qubits, because there exist states that require exponentially large circuits to construct~\cite[Section 4.5]{NC10}.

The following algorithm improves on the trivial algorithm by running in polynomial time, but at the expense of requiring a super-constant number of queries:

\begin{thm}[\cite{Aar16,GR02,KM01,Zal98}] \label{thm:aar}
	There is a uniform sequence $(C_n)_n$ of $\poly(n)$-size quantum circuits, each making $O(n)$ queries to a classical oracle, such that for every $n$-qubit state $\ket\psi$ there exists a classical oracle $f$ such that $C_n^f$ cleanly constructs $\ket\psi$ to within exponentially small error.
\end{thm}

The following algorithm also improves on the trivial algorithm, by running in polynomial rather than exponential \emph{space}, \emph{without} an increase in the number of queries:

\begin{thm}[{Irani et al.~\cite[Theorems 1.3 and 1.4]{INN+22}}] \label{thm:inn}
	There is a nonuniform sequence $(C_n)_n$ of $\poly(n)$-qubit quantum circuits, each making one (resp.\ two) queries to a classical oracle, such that for every $n$-qubit state $\ket\psi$ there exists a classical oracle $f$ such that $C_n^f$ non-cleanly (resp.\ cleanly) constructs $\ket\psi$ to within polynomially (resp.\ exponentially) small error.
\end{thm}

However \cref{thm:inn} does not give an upper bound on the circuit size required to implement the non-query operations, besides the trivial exponential upper bound, and these circuits are nonuniform. Furthermore in the one-query version of \cref{thm:inn}, the approximation error is inverse polynomial rather than inverse exponential.

\begin{figure}
	\centering
	\begin{tabular}{c|c|c|c|c|c|c}
		Algorithm & Queries & Size & Space & Error & Uniform & Clean \\ \hline \hline
		\mr{Trivial} & 1 & \mr{exp} & \mr{exp} & \mr{1/exp} & \mr{yes} & no \\ \cline{2-2} \cline{7-7}
		& 2 & & & & & yes \\ \hline
		\cref{thm:aar} & poly & poly & poly & 1/exp & yes & yes \\ \hline
		\mr{\cref{thm:inn}} & 1 & \mr{exp} & \mr{poly} & 1/poly & \mr{no} & no \\ \cline{2-2} \cline{5-5} \cline{7-7}
		& 2 & & & 1/exp & & yes \\ \hline
		\mr{\makecell{\cref{thm:mi} \\ (this paper)}} & 1 & \mr{poly} & \mr{poly} & \mr {1/exp} & \mr{yes} & no \\ \cline{2-2} \cline{7-7}
		& 4 & & & & & yes
	\end{tabular}
	\caption{Comparison of state synthesis algorithms.}
	\label{fig:comp}
\end{figure}

\cref{fig:comp} compares these algorithms, of which only ours runs in polynomial time using a constant number of queries. Our result answers questions posed by Aaronson~\cite[Question 3.3.6]{Aar16} and Irani et al.~\cite[Section 7]{INN+22}, who collectively asked whether there exists a polynomial-time one-query state synthesis algorithm with exponentially small error.

\subsection{$\cc{statePSPACE} = \cc{stateQIP}(6)$} \label{sec:iss}

Rosenthal and Yuen~\cite{RY21} introduced a notion of interactive proofs for constructing a state $\ket\psi$, where a polynomial-time quantum verifier interacts with an unbounded-complexity but untrusted prover. At the end of the interaction the verifier accepts or rejects, and when accepting the verifier also outputs a state. The \emph{completeness} condition is that there should exist a prover such that the verifier accepts with probability 1. The \emph{soundness} condition is that for every prover such that the verifier accepts with non-negligible probability, the verifier's output state conditioned on accepting should be an approximation of $\kb\psi$.

Rosenthal and Yuen~\cite{RY21} defined $\cc{stateQIP}$ as the class of state sequences that can be constructed in this way, in analogy with the class $\cc{QIP}$ of decision problems with similar quantum interactive protocols. They also defined $\cc{statePSPACE}$ as a quantum state analogue of $\cc{PSPACE}$. (Formal definitions are given in \cref{sec:scc}.) Then Rosenthal and Yuen proved the inclusion $\cc{statePSPACE} \subseteq \cc{stateQIP}$, and Metger and Yuen~\cite{MY23} proved the converse inclusion $\cc{stateQIP} \subseteq \cc{statePSPACE}$. This establishes the equality $\cc{stateQIP} = \cc{statePSPACE}$, a quantum state analogue of $\cc{QIP} = \cc{PSPACE}$~\cite{jain2011qip}, which is itself a quantum analogue of $\cc{IP} = \cc{PSPACE}$~\cite{lund1992algebraic,shamir1992ip}.

Rosenthal and Yuen's proof that $\cc{statePSPACE} \subseteq \cc{stateQIP}$ goes roughly as follows. Let $\ket\psi$ denote the $n$-qubit state that the verifier would like to construct, and let $f$ be the oracle associated with constructing $\ket\psi$ in \cref{thm:aar}. Tomography of states in $\cc{statePSPACE}$ can be done in $\cc{PSPACE}$ since $\cc{PSPACE} = \cc{BQPSPACE}$~\cite{watrous03complexity}, and inspection of the proof of \cref{thm:aar} reveals that $f$ can be computed in $\cc{PSPACE}$ given query access to the description of $\ket\psi$. Therefore $f$ can be computed in $\PSPACE$, which suggests the following candidate protocol for constructing $\ket\psi$: simulate the algorithm from \cref{thm:aar}, with queries to $f$ answered by running the $\IP = \PSPACE$ protocol in superposition.

However, controlled on an input string $x$ to the $\IP = \PSPACE$ protocol for $f$, there is a garbage state associated with $x$ at the end of the $\IP = \PSPACE$ protocol. The prover is required to help the verifier uncompute this garbage state, so that the verifier's output register is not entangled with the rest of the system. The main challenge is to ensure that the prover uncomputes this garbage state honestly, which the verifier achieves using an intricate sequence of swap tests. Finally the soundness of the protocol is improved by repeating the above procedure polynomially many times, accepting if and only if every instance accepts, and then outputting the output state of a random instance.

Rosenthal and Yuen~\cite{RY21} posed the question of whether there exists a $\cc{statePSPACE} \subseteq \cc{stateQIP}$ protocol with a constant number of rounds of interaction. Since the $\PSPACE \subseteq \cc{QIP}$ protocol can be parallelized to three total messages~\cite{watrous03pspace}, the main obstacle is that the algorithm from \cref{thm:aar} makes a super-constant number of queries. A second, more subtle obstacle is that Rosenthal and Yuen's proof of correctness of the above soundness amplification procedure breaks down if the instances are run in parallel.

Using the one-query version of \cref{thm:mi} we prove the following, where $\cc{stateQIP}(6)$ is defined similarly to $\cc{stateQIP}$ but for protocols with six total messages:

\begin{thm} \label{thm:spq}
	$\cc{statePSPACE} \subseteq \cc{stateQIP}(6)$.
\end{thm}

\cref{thm:spq} mostly follows by substituting \cref{thm:mi} for \cref{thm:aar} in Rosenthal and Yuen's~\cite{RY21} proof that $\cc{statePSPACE} \subseteq \cc{stateQIP}$. However, we present a self-contained proof of \cref{thm:spq} for three reasons. First, it is necessary to prove that a soundness amplification procedure similar to the one used by Rosenthal and Yuen can be parallelized. Second, we can substitute a single, easily defined projective measurement for the sequence of swap tests in Rosenthal and Yuen's proof, using the fact that the queries in our one-query state synthesis algorithm are (trivially) non-adaptive. This significantly simplifies the description of the verifier and the proof of soundness. Third, in their proof of the converse inclusion $\cc{stateQIP} \subseteq \cc{statePSPACE}$, Metger and Yuen~\cite{MY23} used definitions of $\cc{statePSPACE}$ and $\cc{stateQIP}$ slightly different than those of Rosenthal and Yuen. We adopt Metger and Yuen's definitions in our proof of \cref{thm:spq}, implying the equality $\cc{statePSPACE} = \cc{stateQIP} = \cc{stateQIP}(6)$.

\subsection{Barrier to \qaczf lower bounds for constructing explicit states} \label{sec:bar}

In classical circuit complexity it is notoriously difficult to prove that an explicit Boolean function is hard for a given circuit class. The same holds for quantum circuit complexity, since quantum circuits can simulate Boolean circuits. However this does not immediately imply that it should be difficult to prove quantum circuit lower bounds for \emph{quantum} tasks with no classical analogue, such as constructing a quantum state. And in fact, Jia and Wolf~\cite{JW23} found explicit states that require exponential circuit size to \emph{exactly} construct.

Nevertheless, Aaronson~\cite{Aar16} observed a barrier to finding explicit states that cannot be \emph{approximately} constructed by $\cc{BQP/poly}$ circuits (i.e.\ nonuniform polynomial-size quantum circuits) to within exponentially small error. Specifically, let $\ket{\psi_n}$ be an $n$-qubit state for all $n$ and let $f_n$ be the oracle associated with approximately constructing $\ket{\psi_n}$ in \cref{thm:aar}. If $(f_n)_n$ can be computed in $\cc{BQP/poly}$, then plugging these circuits for $(f_n)_n$ into the algorithm from \cref{thm:aar} yields a sequence of $\cc{BQP/poly}$ circuits for approximately constructing $(\ket{\psi_n})_n$. Conversely, if there are no $\cc{BQP/poly}$ circuits for approximately constructing $(\ket{\psi_n})_n$ then there are no $\cc{BQP/poly}$ circuits for computing $(f_n)_n$. This would be a breakthrough result, since finding an explicit function that is not in $\cc{BQP/poly}$ (or even $\cc{P/poly}$) is a longstanding open problem.

However this still leaves open the possibility of finding explicit states that cannot be approximately constructed by $\mc C$ circuits, for some nonuniform quantum circuit class $\mc C$ that (as far as we know) is weaker than $\cc{BQP/poly}$. One such class is \qaczf, a quantum analogue of \acz introduced by Green, Homer, Moore and Pollett~\cite{Gre+02} which we define in \cref{sec:prelim}. Analogously to \acz, one motivation for proving lower bounds against \qaczf is that it is contained in \qnco (i.e.\ log-depth circuits with one- and two-qubit gates), and another motivation is that \qaczf is one of the weakest quantum circuit classes that is natural to define. The ``next weakest" class \qacz~\cite{Gre+02} is like \qaczf except without ``fanout gates" that make copies of a classical bit, and the even weaker class $\cc{QNC^0}$ is easy to prove lower bounds against by light cone arguments.

The non-query operations from \cref{thm:mi} can be efficiently implemented in \qaczf, so we can rule out this possibility by reasoning similar to that in Aaronson's barrier:

\begin{obs} \label{obs:bar}
	\qaczf lower bounds for cleanly constructing explicit states (to within exponentially small error) would imply \qaczf lower bounds for computing explicit Boolean functions.
\end{obs}

We state \cref{obs:bar} more formally in \cref{sec:qacfz}. It is known that $\cc{TC^0} \subseteq \cc{QAC_f^0}$~\cite{HS05,TT16}, where \tcz denotes the class of functions computable by non-uniform polynomial-size Boolean circuits with NOT gates and unbounded-fanin AND, OR, and MAJORITY gates. It is an open problem to prove superpolynomial-size $\cc{TC^0}$ lower bounds for an explicit Boolean function, so \cref{obs:bar} implies a barrier to proving superpolynomial-size \qaczf lower bounds for approximately constructing explicit states.

\begin{rmk*}
	Another consequence of the fact that the non-query operations in \cref{thm:mi} can be implemented in \qaczf is that, by simulating the queries with CNF or DNF formulas, exponentially large \qaczf circuits can approximately construct any state. This was previously proved by the author~\cite{R21b}, in fact for exact constructions.
\end{rmk*}

\subsection{Approximately constructing arbitrary states} \label{sec:cub}

Every $n$-qubit pure state can be cleanly, exactly constructed with $O(2^n)$ one- and two-qubit gates~\cite{Gui+23,Sun+21,YZ23,ZLY22}. This upper bound is tight, because by a dimension-counting argument a Haar random state almost surely requires $\Omega(2^n)$ one- and two-qubit gates to construct \emph{exactly}. However we show that circuits of size $o(2^n)$ can \emph{approximately} construct any $n$-qubit state, moreover with gates from an appropriate finite gate set:

\begin{restatable}{thm}{rubs} \label{thm:ubs}
	There exists a finite gate set $\mc G$ such that for all $n \in \N, \eps \ge \exp(-\poly(n))$ and $n$-qubit states $\ket\psi$, there exists a circuit $C$ consisting of $O(2^n \log(1/\eps) / n)$ gates from $\mc G$ such that $\Norm{C\ket\zs - \ket\psi \ket\zs} \le \eps$.
\end{restatable}

We prove \cref{thm:ubs} using Lupanov's~\cite{Lup58} $O(2^m/m)$-size Boolean circuit for an arbitrary function $f: \cube m \to \bits$, applied to the oracle from the clean version of \cref{thm:mi}. The portion of the circuit corresponding to the non-query operations is converted to a circuit over $\mc G$ using the Solovay--Kitaev theorem~\cite{BG21,DN06}. We require $\eps \ge \exp(-\poly(n))$ for convenience, but our proof technique implies a similar statement for smaller $\eps$ as well.

The circuit from \cref{thm:ubs} uses exponentially many ancillae. Some ancillae are necessary, because Nielsen and Chuang~\cite[Section 4.5.4]{NC10} proved by a counting argument that without ancillae, for every finite gate set $\mc G$ there exist states that require $\Omega\Paren{2^n \log(1/\eps) / \log n}$ gates from $\mc G$ to construct to within error $\eps$. We also prove an analogue of Nielsen and Chuang's lower bound for non-clean constructions \emph{with} ancillae, by a similar counting argument:

\begin{restatable}{thm}{rlbs} \label{thm:lbs}
	Let $\mc G$ be a finite gate set. Then for all $n \in \N$ and $1/4 \ge \eps \ge \exp(-\poly(n))$, there exists an $n$-qubit state $\ket\psi$ such that circuits $C$ over $\mc G$ require $\Omega\Paren{2^n \log(1/\eps) / n}$ gates in order for the reduced state $\rho$ on the first $n$ qubits of $C \ket\zs$ to satisfy $\td\Paren{\rho, \kb\psi} \le \eps$.
\end{restatable}

To properly compare \cref{thm:ubs,thm:lbs} it is necessary to convert the error bound in \cref{thm:ubs} from 2-norm error to trace distance error. Identifying a pure state $\ket\phi$ with the density matrix $\kb\phi$, the trace distance between two pure states is at most the 2-norm distance between those states (see \cref{eq:tdps}), so the conclusion of \cref{thm:ubs} implies that the trace distance between $\ket\psi \ket\zs$ and $C \ket\zs$ is at most $\eps$. Therefore the trace distance between $\ket\psi$ and the reduced state on the first $n$ qubits of $C \ket\zs$ is at most $\eps$, so the lower bound from \cref{thm:lbs} matches the upper bound from \cref{thm:ubs}.

It is an open problem whether the $\eps \ge 2^{-O(n)}$ case of \cref{thm:lbs} generalizes to circuits consisting of arbitrary one- and two-qubit gates. (The $\eps \le 2^{-\omega(n)}$ case cannot admit such a generalization, by the previously mentioned $O(2^n)$ upper bounds for exact constructions with arbitrary one- and two-qubit gates.) A slightly weaker lower bound of $2^n / \poly(n)$ holds for such circuits, because by the Solovay--Kitaev theorem~\cite{BG21,DN06} circuits consisting of $2^n n^{-\omega(1)}$ one- and two-qubit gates can be simulated to within exponentially small error by circuits consisting of $2^n n^{-\omega(1)}$ gates from a universal gate set, and therefore by \cref{thm:lbs} cannot construct arbitrary $n$-qubit states to within error $\eps$.\footnote{A tighter analysis can be obtained using Harrow, Recht and Chuang's~\cite{HRC02} version of the Solovay--Kitaev theorem, which says that for certain finite gate sets, any unitary on a fixed number of qubits can be approximated to within error $\eps$ in the operator 2-norm by $O(\log(1/\eps))$ (rather than $\poly \log(1/\eps)$) gates.}

\subsection{Organization} \label{sec:org}

\cref{sec:prelim} is the preliminaries. In \cref{sec:post} we prove a weaker variant of the one-query version of \cref{thm:mi}, where the circuit postselects on a certain measurement outcome that occurs with constant probability. By reducing to this result in different ways, in \cref{sec:msa} we prove the one- and four-query versions of \cref{thm:mi}. In \cref{sec:scc} we define $\cc{statePSPACE}$ and $\cc{stateQIP}(6)$ and introduce other related background, in preparation for the proof in \cref{sec:psq} that $\cc{statePSPACE} \subseteq \cc{stateQIP}(6)$ (i.e.\ \cref{thm:spq}). In \cref{sec:qacfz} we state and prove \cref{obs:bar} more formally, and in \cref{sec:states} we prove \cref{thm:ubs,thm:lbs}.

\section{Preliminaries} \label{sec:prelim}
Logarithms in this paper are base 2. We write $(x_n)_n$ to denote the infinite sequence $(x_1, x_2, \dotsc)$ for some class of objects $x_n$.
\paragraph*{Space-bounded computation}

All Turing machines in this paper have a read-only input tape, read-write work tapes, and a write-only output tape. Let $\cube*$ denote the set of finite strings over $\bits$, and for $x \in \cube*$ let $|x|$ denote the length of $x$. For $s: \N \to \N$ a Turing machine $M$ uses space $s$ if for all $x \in \cube*$, at most $s(|x|)$ cells are used on the work tapes in the computation of $M(x)$. If $M$ uses space $s$ and halts then $M$ uses time $O(2^s)$, so $|M(x)| \le O\Paren{2^{s(|x|)}}$ for all $x$.

\paragraph*{Description of a pure state}

For $\ket\psi = \sum_{x \in \cube n} \alpha_x \ket x$ and $\eps>0$, we define an \emph{$\eps$-precision description of $\ket\psi$} to be a tuple $\Paren{\tilde\alpha_x}_{x \in \cube n}$ of complex numbers specified exactly in binary such that $\Mag{\tilde\alpha_x - \alpha_x} \le \eps$ for all $x$. We will often leave $\eps$ implicit and simply refer to ``the description of $\ket\psi$", by which we mean an $\exp(-p(n))$-precision description of $\ket\psi$ where $p$ is a polynomial that may be chosen to be as large as desired; in this case $\poly(n)$ bits of precision are needed to specify $\tilde\alpha_x$.

\paragraph*{Quantum information theory}

A \emph{register} $\reg R$ is a named finite-dimensional complex Hilbert space. If $\reg A, \reg B, \reg C$ are registers, for example, then the concatenation $\reg{ABC}$ denotes the tensor product of the associated Hilbert spaces. For a linear transformation $L$ and register $\reg R$, we write $L_{\reg R}$ to indicate that $L$ acts on $\reg R$, and similarly we write $\rho_{\reg R}$ to indicate that a state $\rho$ is in the register $\reg R$. We write $\tr\cdot$ to denote trace, $\Tr_{\reg R}(\cdot)$ to denote the partial trace over a register $\reg R$, and $\trg n \cdot$ to denote the partial trace over all but the first $n$ qubits. We write $I_n$ to denote the $n$-qubit identity transformation, or $I$ when the number of qubits is implicit. We write $\ket+ = \frac {\ket0 + \ket1} {\sqrt 2}, \ket- = \frac {\ket0 - \ket1} {\sqrt2}$ to denote the Hadamard basis states, and $\crl U = \kb0 \otimes I + \kb1 \otimes U$ to denote controlled-$U$. For a (not necessarily normalized) vector $\ket\psi$ we write $\psi = \kb\psi$. We write $\norm\cdot$ to denote the vector 2-norm.

Let $\norm{M}_1 = \tr{|M|}$ denote the trace norm of a matrix $M$, and let $\td(\rho, \sigma) = \frac12 \norm{\rho - \sigma}_1$ denote the trace distance between mixed states $\rho$ and $\sigma$. We use the fact that
\begin{equation} \label{eq:tdc}
	\td\Paren{\Phi(\rho), \Phi(\sigma)} \le \td(\rho, \sigma)
\end{equation}
for all channels $\Phi$ and states $\rho, \sigma$. We also use the following special case of the Fuchs-van de Graaf inequality: if $\rho$ is a mixed state and $\ket\psi$ is a pure state then
\begin{equation} \label{eq:fvdg}
	\td(\rho, \psi) \le \sqrt{1 - \tr{\rho \psi}} = \sqrt{\tr{\rho(I - \psi)}}.
\end{equation}
(See e.g.\ Nielsen and Chuang~\cite[Chapter 9]{NC10} for proofs of \cref{eq:tdc,eq:fvdg}.) In particular, if $\ket\psi$ and $\ket\phi$ are pure states then
\begin{equation} \label{eq:tdps}
	\td\Paren{\psi, \phi}
	\le \sqrt{1 - \Mag{\ip \psi \phi}^2}
	= \sqrt{\Paren{1 + \Mag{\ip \psi \phi}} \Paren{1 - \Mag{\ip \psi \phi}}}
	\le \sqrt{2\Paren{1 - \mrm{Re}\Paren{\ip \psi\phi}}} \\
	= \Norm{\ket\psi - \ket\phi}.
\end{equation}

\paragraph*{Quantum query models}

By a \emph{quantum circuit making $k$ queries to an $n$-qubit quantum oracle and its inverse}, we mean a circuit of the form $C = C_k Q_k C_{k-1} Q_{k-1} \dotsb C_0$ where each $C_j$ is a unitary and each $Q_j$ is a placeholder for either a ``forward" or ``backward" query. For an $n$-qubit unitary $A$, by $C^A$ we mean the unitary defined by substituting $A$ and $\adj A$ respectively for the forward and backward queries in $C$. Claims about the quantum circuit complexity of $C$ are in reference to the circuit $C_k C_{k-1} \dotsb C_0$ defined by removing the queries from $C$.

The following result of the author~\cite{R21b} says that if the task of constructing a state $\ket\psi$ reduces to that of constructing a state $\ket\phi$, then the task of \emph{approximately} constructing $\ket\psi$ reduces to that of \emph{approximately} constructing $\ket\phi$:

\begin{lem}[{special case\protect\footnote{Specifically, the case where their $J$ equals our $\ket\psi$, their $A$ equals our $\kb0 \otimes I_n + \kb1 \otimes \ketbra \phi {0^n}$, their $U$ equals our $\crl U$, their $B$ equals our $\kb0 \otimes I_n + \kb1 \otimes V \kb{0^n}$, and their $V$ equals our $\crl V$.} of \cite[Lemma 3.3]{R21b}}] \label{lem:err}
	Let $C$ be an $m$-qubit quantum circuit making $k$ queries to an $(n+1)$-qubit quantum oracle and its inverse, and let $\ket\psi$ be an $m$-qubit state. Assume there exists an $n$-qubit state $\ket\phi$ such that for all $n$-qubit unitaries $U$ satisfying $U \ket{0^n} = \ket\phi$, it holds that $C^{\crl U} \ket{0^m} = \ket\psi$. Then for all $n$-qubit unitaries $V$ it holds that $\Norm{C^{\crl V} \ket{0^m} - \ket\psi} \le \sqrt 2 \cdot k \cdot \norm{\ket\phi - V \ket{0^n}}$.
\end{lem}

Queries to a classical oracle (i.e.\ a Boolean function) can be modeled in either of two standard ways. In the first, a function $f: \cube n \mapsto \cube m$ is encoded as the oracle $U_f$ defined by $U_f \ket{x, y} = \ket{x, y \oplus f(x)}$. In the second, which is only applicable when $m=1$, the function $f$ is instead encoded as the oracle $V_f$ defined by $V_f \ket{x} = (-1)^{f(x)} \ket{x}$. These models are equivalent, because $V_f = (I_n \otimes \bra-) U_f (I_n \otimes \ket-)$, and if $g(x,y) = \bigoplus_{j=1}^m f(x)_j y_j$ (where the subscript $j$ indicates the $j$'th bit of an $m$-bit string) then $U_f = (I_n \otimes H^{\otimes m}) V_g (I_n \otimes H^{\otimes m})$ where $H$ denotes the Hadamard gate~\cite{BV97,NC10}. We write $C^f$ to abbreviate $C^{U_f}$ or $C^{V_f}$; since $U_f$ and $V_f$ are Hermitian we do not need to distinguish between forward and backward queries to a classical oracle.

We use the fact that parallel queries to classical oracles can be merged into a single query to a classical oracle, i.e.\
\begin{align} \label{eq:par-quer}
	&V_{f_1} \otimes \dotsb \otimes V_{f_k} = V_F
	&\text{for}&
	& F\Paren{x^{(1)}, \dotsc, x^{(k)}} = \bigoplus_{j=1}^k f_j\Paren{x^{(j)}}
\end{align}
for all functions $f_1, \dotsc, f_k$. More generally, a collection of parallel queries of the form $\bigotimes_j U_{f_j} \otimes \bigotimes_k V_{g_k}$ can be merged into a single query to a classical oracle, using the above equivalence between the query models.

\paragraph*{\qaczf circuits}

A \qaczf circuit~\cite{Gre+02} is a constant-depth quantum circuit consisting of arbitrary one-qubit gates, as well as \emph{generalized Toffoli gates} of arbitrary arity defined by
\begin{equation*}
	\ket{b,x} \mapsto \ket{b \oplus \prod_{j=1}^n x_j, x} \quad \text{for} \quad b \in \bits, x = (x_1, \dotsc, x_n) \in \cube n,
\end{equation*}
and \emph{fanout gates} of arbitrary arity defined by
\begin{equation*}
	\ket{b,x} \mapsto \ket{b, x \oplus b^n} \quad \text{for} \quad b \in \bits, x \in \cube n.
\end{equation*}

The following results of the author~\cite{R21b} are easy to prove:

\begin{lem}[{\cite[Lemma 4.3]{R21b}}] \label{lem:qaczf-swap}
	There is a uniform family of $O\Paren{m n \log n}$-qubit \qaczf circuits $(C_{n,m})_{n,m}$, where $C_{n,m}$ takes as input a $(\log n)$-qubit register $\reg K$ and $m$-qubit registers $\reg A_0, \dotsc, \reg A_{n-1}, \reg B$ (and ancillae) and swaps $\reg A_k$ and $\reg B$ controlled on the classical state $\ket{k}_{\reg K}$.
\end{lem}

\begin{lem}[{special case of \cite[Lemma 4.4]{R21b}}] \label{lem:cqacz}
	If $U$ is an $n$-qubit \qaczf circuit then there exists an $O(n)$-qubit \qaczf circuit $C$ such that $C(I_{n+1} \otimes \ket\zs) = \crl U \otimes \ket\zs$.
\end{lem}

The following lemma says that $\cc{QAC_f^0} \subseteq \cc{QNC^1}$:

\begin{lem}[{folklore, or special case of \cite[Lemma A.1]{R21b}}] \label{lem:qsim}
	For all $n$-qubit \qaczf circuits $U$, there exists an $O(n)$-qubit, $O(\log n)$-depth circuit $C$ consisting of $O(n)$ one- and two-qubit gates such that $C(I_n \otimes \ket\zs) = U \otimes \ket\zs$.
\end{lem}

\section{One-query state synthesis with postselection} \label{sec:post}
In this section we prove the following lemma, which will be used in our proof of \cref{thm:mi}. This lemma implies an efficient one-query state synthesis algorithm given the ability to postselect on a measurement outcome that occurs with approximately constant probability:

\begin{lem} \label{lem:help}
	There is a real number $\gamma \approx 0.18$ such that the following holds. Let $\eps: \N \to (0,1/2)$ be a function such that $\eps(n) \ge \exp(-\poly(n))$ and $\eps(n)$ is computable in $\poly(n)$ time for all $n$, and let $t(n) = \ceil{\log\log(1/\eps(n))} + 7$. Then there is a uniform sequence of $\poly(n)$-qubit \qaczf circuits $(A_n)_n$, each making one query to a classical oracle, such that for every $n$-qubit state $\ket\psi$ there exists a classical oracle $f = f_{\ket\psi}$, a $(t(n)+n)$-qubit state $\ket\tau$ such that $\Paren{\bra{0^{t(n)}} \otimes I_n} \ket\tau = 0$, and a string $z \in \cube{\poly(n)}$ such that
	\begin{equation} \label{eq:help}
		\Norm{A_n^f \ket\zs - \Paren{\gamma \ket{0^{t(n)}} \ket\psi + \sqrt{1 - \gamma^2} \ket\tau} \ket{z}} \le \eps(n).
	\end{equation}
	Furthermore there is an algorithm that takes as input the description of an $n$-qubit state $\ket\psi$ and a string $x$, runs in $\poly(n)$ space, and outputs $f_{\ket\psi} (x)$.
\end{lem}

The proof is organized as follows. In \cref{sec:halg} we describe $A_n$ and $f$, and in \cref{sec:hpc} we prove that \cref{eq:help} holds. In \cref{sec:qps} we prove that $f$ can be computed in $\poly(n)$ space; actually we prove this for a slightly different oracle $f^\prime$ due to a subtlety involving floating-point arithmetic, but we show that \cref{eq:help} also holds with $f^\prime$ in place of $f$ (and with slightly different values of $\ket\tau$ and $z$).
\subsection{The algorithm} \label{sec:halg}

Our algorithm uses \emph{Clifford unitaries}, which are products of Hadamard, phase, and CNOT gates, i.e.\ products of the gates
\begin{align*}
	&H = \frac1{\sqrt 2} \begin{pmatrix} 1 & 1 \\ 1 & -1 \end{pmatrix},&
	&S = \begin{pmatrix} 1 & \\ & i \end{pmatrix},&
	&\mathit{CNOT} = \begin{pmatrix} 1 & & & \\ & 1 & & \\ & & & 1 \\ & & 1 & \end{pmatrix}.
\end{align*}
Let
\begin{equation*}
	\alpha = 0.35, \qquad
	\beta = \sqrt{1-\alpha^2} \approx 0.94, \qquad
	\gamma = (1-\beta)/\alpha \approx 0.18.
\end{equation*}
For a complex number $c$, let $\sr(c) = 1$ if the real part of $c$ is nonnegative, and let $\sr(c) = -1$ otherwise. For a vector $\ket\eta \in \Paren{\C^2}^{\otimes n}$ and a Clifford unitary $C$, let
\begin{equation*}
	\ket{p_{\eta, C}} = C \cdot 2^{-n/2} \sum_{\mathclap{x \in \cube n}} \sr(\bra\eta C \ket x) \ket x.
\end{equation*}
The following is implicit in Irani et al.~\cite{INN+22}, as explained in \cref{app:a5}:

\begin{restatable}[{\cite{INN+22}}]{lem}{raf} \label{lem:a5}
	For all states $\ket\eta$ there exists a Clifford unitary $C$ such that $\mrm{Re}(\ip \eta {p_{\eta, C}}) \ge \alpha$.
\end{restatable}

\begin{rmk*}
	In \cref{app:var} we prove an analogue of \cref{lem:a5} for a class of states other than $\ket{p_{\eta, C}}$, which can be used to give an alternate proof of a statement similar to \cref{lem:help}. The idea to use $\ket{p_{\eta, C}}$ was suggested to us by Fermi Ma~\cite{Ma23} after we sketched the argument in \cref{app:var} to him.
\end{rmk*}

Fix $n, \ket\psi, t = t(n)$ as in \cref{lem:help}. For $k \ge 0$, given states $\ket{\phi_0}, \dotsc, \ket{\phi_{k-1}}$, let $\ket{\eta_k} = \ket\psi - \alpha \sum_{j=0}^{k-1} \beta^j \ket{\phi_j}$, and (using \cref{lem:a5}) let $C_k$ be a Clifford unitary such that the state $\ket{\phi_k} = \ket{p_{\eta_k, C_k}}$ satisfies $\mrm{Re} (\ip {\eta_k} {\phi_k}) \ge \alpha \norm{\ket{\eta_k}}$.

Let $T = 2^t$ and $\ket\sigma = \sqrt{\frac{1 - \beta}{1 - \beta^T}} \cdot \sum_{j=0}^{T-1} \sqrt{\beta^j} \ket{j}$. Observe that
\begin{equation*}
	\ket\sigma = \sqrt{\frac{1-\beta}{1-\beta^T}}
	\Paren{\ket0 + \beta^{2^{t-2}} \ket1} \otimes
	\Paren{\ket0 + \beta^{2^{t-1}} \ket1} \otimes \dotsb
	\Paren{\ket0 + \beta \ket1} \otimes
	\Paren{\ket0 + \beta^{1/2} \ket1},
\end{equation*}
so there is a tensor product $L$ of $t$ one-qubit gates such that $L\ket{0^t} = \ket{\sigma}$.

The circuit $A_n$  is described in \cref{alg:oss}, where $\reg A$ is a $t$-qubit register and $\reg B$ is an $n$-qubit register. Although the algorithm is phrased in terms of multiple queries, these can be merged into a single query using \cref{eq:par-quer} and the surrounding discussion. Aaronson and Gottesman~\cite[Theorem 8]{AG04} proved that every Clifford unitary can be written as a round of Hadamard gates, then a round of CNOT gates, then a round of phase gates, and so on in the sequence H-C-P-C-P-C-H-P-C-P-C with no ancillae (a ``round" may consist of any number of layers of the given gate type). On \cref{line:o7,line:o8}, by the description of a Clifford unitary we mean the concatenation of the descriptions of the rounds comprising that unitary, defined as follows:
\begin{itemize}
	\item Since $H^2 = I$ a round of Hadamard gates equals $\bigotimes_{j=1}^n H^{x_j}$ for some string $x = (x_1, \dotsc, x_n) \in \cube n$. Call $x$ the description of this round.
	\item Similarly since $S^4 = I$, a round of phase gates can be described by a string in $\{0,1,2,3\}^n$.
	\item A round of CNOT gates acts on the standard basis as $\ket x \mapsto \ket{Mx}$ for some $M \in \mrm{GL}_n(\ft)$, because this holds for a single CNOT gate and $\mrm{GL}_n(\ft)$ is closed under multiplication. Call the pair $(M, M^{-1})$ the description of this round.
\end{itemize}

\begin{algorithm}
	\caption{Circuit and oracle for \cref{lem:help}}
	\label{alg:oss}
	\begin{algorithmic}[1]
		\State Construct $\ket\sigma_{\reg A} \ket{+^n}_{\reg B}$. \Comment Using $L$.
		\Ctrl{the classical state $\ket{j}_{\reg A} \ket{x}_{\reg B}$,}
		\State apply a phase of $\sr(\bra{\eta_j} C_j \ket x)$ by querying the oracle. \hlabel{line:o3}
		\EndCtrl
		\State Query descriptions of $C_0, \dotsc, C_{T-1}$. \Comment Merge with the \cref{line:o3} query using \cref{eq:par-quer}. \hlabel{line:o7}
		\Ctrl{the classical state $\ket{j}_{\reg A}$,}
		\State Apply $(C_j)_{\reg B}$ using the queried description of $C_j$. \hlabel{line:o8}
		\EndCtrl \hlabel{line:o9}
		\State Apply $\adj L_{\reg A}$.
	\end{algorithmic}
\end{algorithm}

We now describe the \qaczf implementation of \cref{line:o8} in greater detail. Given an index $j$ and descriptions of Clifford unitaries $C_0, \dotsc, C_{T-1}$, the description of $C_j$ can be computed using \cref{lem:qaczf-swap}. A polynomial-size \qaczf circuit can then implement $C_j$ by successively implementing the rounds comprising $C_j$. Rounds of Hadamard and phase gates can be implemented trivially. To implement a round of CNOT gates acting as $\ket x \mapsto \ket{Mx}$, first compute $y = Mx$, and then uncompute $x = M^{-1} y$ controlled on $y$, using that parity is in \qaczf~\cite{Gre+02}. Since $\eps(n) \ge \exp(-\poly(n))$ it holds that $t \le O(\log n)$ and $T \le \poly(n)$, so $A_n$ requires $\poly(n)$ qubits.

\begin{rmk*}
	Aaronson and Gottesman~\cite[Section 6]{AG04} used similar reasoning to prove that Clifford unitaries can be implemented in \qnco. The purpose of querying the description of $C_j$ in \cref{line:o7}, rather than applying it nonuniformly in \cref{line:o8}, is to make the circuit (unlike the oracle) independent of $\ket\psi$. The purpose of querying \emph{all} of $C_0, \dotsc, C_{T-1}$, rather than just $C_j$, is so that the register holding the description of $C_j$ is unentangled with the rest of the system.
\end{rmk*}

\subsection{Proof of correctness} \label{sec:hpc}

The string $z$ referred to in the lemma is the concatenation of the descriptions of $C_0, \dotsc, C_{T-1}$ along with some number of zeros. Let $\ket\varphi$ denote the final state in $\reg{AB}$, let $\ket\theta = \bra{0^t}_{\reg A} \ket\varphi$, and let
\begin{equation*}
	\ket\tau
	= \frac{\Paren{I - \kb{0^t}}_{\reg A} \ket\varphi} {\Norm{\Paren{I - \kb{0^t}}_{\reg A} \ket\varphi}}
	= \frac{\Paren{I - \kb{0^t}}_{\reg A} \ket\varphi} {\sqrt{1 - \norm{\ket\theta}^2}}.
\end{equation*}
Then
\begin{align*}
	&\Norm{A_n^f \ket\zs - \Paren{\gamma \ket{0^t} \ket\psi + \sqrt{1 - \gamma^2} \ket\tau} \ket{z}}^2 \\
	&= \Norm{\ket\varphi - \Paren{\gamma \ket{0^t} \ket\psi + \sqrt{1 - \gamma^2} \ket\tau}}^2 \\
	&= \Norm{\ket{0^t} \Paren{\ket\theta - \gamma \ket\psi} + \Paren{\sqrt{1 - \norm{\ket\theta}^2} - \sqrt{1 - \gamma^2}} \ket\tau}^2 \\
	&= \norm{\ket\theta - \gamma \ket\psi}^2 + \Paren{\sqrt{1 - \norm{\ket\theta}^2} - \sqrt{1 - \gamma^2}}^2 \\
	&= \norm{\ket\theta - \gamma \ket\psi}^2 + \Paren{\frac{\Paren{\norm{\ket\theta} + \gamma} \Paren{\norm{\ket\theta} - \gamma}} {\sqrt{1 - \norm{\ket\theta}^2} + \sqrt{1 - \gamma^2}}}^2 \\
	&\le \norm{\ket\theta - \gamma \ket\psi}^2 + \frac{(1 + \gamma)^2} {1 - \gamma^2} \Paren{\norm{\ket\theta} - \gamma}^2 &\text{because $\norm{\ket\theta} \le 1$} \\
	&\le \Paren{1 + \frac{(1 + \gamma)^2} {1-\gamma^2}} \norm{\ket\theta - \gamma \ket\psi}^2 &\text{by the triangle inequality} \\
	&\le 2.45 \Norm{\ket\theta - \gamma \ket\psi}^2,
\end{align*}
so
\begin{equation*}
	\Norm{A_n^f \ket\zs - \Paren{\gamma \ket{0^t} \ket\psi + \sqrt{1 - \gamma^2} \ket\tau} \ket{z}}
	\le 1.57 \Norm{\ket\theta - \gamma \ket\psi}.
\end{equation*}

Inspection of \cref{alg:oss} reveals that
\begin{equation*}
	\ket\varphi
	= \adj L_{\reg A} \Paren{\sum_{j<T}
		\Paren{j_{\reg A} \otimes C_j \sum_{\mathclap{x \in \cube n}} \sr\Paren{\bra{\eta_j} C_j \ket{x}} x_{\reg B}}}
	\ket\sigma_{\reg A} \ket{+^n}_{\reg B},
\end{equation*}
so
\begin{align*}
	\ket\theta
	&= \bra\sigma_{\reg A}
	\Paren{\sum_{j<T} \Paren{j_{\reg A} \otimes C_j \cdot 2^{-n/2} \sum_{\mathclap{x \in \cube n}} \sr\Paren{\bra{\eta_j} C_j \ket{x}} \ket{x}_{\reg B}}}
	\ket\sigma_{\reg A} \\
	&= \sum_{j<T} |\ip{j}\sigma|^2 \ket{p_{\eta_j, C_j}}_{\reg B}
	= \frac{1-\beta}{1-\beta^T} \sum_{j<T} \beta^j \ket{\phi_j}
	= \frac{1-\beta}{1-\beta^T} \cdot \frac{\ket\psi - \ket{\eta_T}} \alpha
	= \frac{\gamma \Paren{\ket\psi - \ket{\eta_T}}} {1-\beta^T},
\end{align*}
and therefore by the triangle inequality
\begin{align*}
	\Norm{\ket\theta - \gamma \ket\psi}
	&= \frac \gamma {1-\beta^T} \Norm{\Paren{\ket\psi - \ket{\eta_T}} - \Paren{1 - \beta^T} \ket\psi}
	= \frac \gamma {1-\beta^T} \Norm{\beta^T \ket\psi - \ket{\eta_T}} \\
	&\le \frac \gamma {1-\beta} \Paren{\beta^T + \Norm{\ket{\eta_T}}}
	\le 2.86 \Paren{\beta^T + \Norm{\ket{\eta_T}}}.
\end{align*}

We prove by induction on $k$ that $\norm{\ket{\eta_k}} \le \beta^k$ for all $k$. The base case $k=0$ holds because $\ket{\eta_0} = \ket\psi$. If the claim holds for $k$, then
\begin{align*}
	\norm{\ket{\eta_{k+1}}}^2
	&= \Norm{\ket{\eta_k} - \alpha \beta^k \ket{\phi_k}}^2
	= \Norm{\ket{\eta_k}}^2 - 2 \alpha \beta^k \mrm{Re}(\ip{\eta_k}{\phi_k}) + \alpha^2 \beta^{2k} \\
	&\le \Norm{\ket{\eta_k}}^2 - 2 \alpha^2 \beta^k \Norm{\ket{\eta_k}} + \alpha^2 \beta^{2k},
\end{align*}
where the inequality is by the definition of $\ket{\phi_k}$. This bound is convex as a function of $\norm{\ket{\eta_k}}$, so it achieves its maximum over $0 \le \norm{\ket{\eta_k}} \le \beta^k$ at either $\norm{\ket{\eta_k}} = 0$ or $\norm{\ket{\eta_k}} = \beta^k$. In both cases it follows straightforwardly that $\norm{\ket{\eta_{k+1}}} \le \beta^{k+1}$, using in the $\norm{\ket{\eta_k}} = 0$ case the fact that $\alpha < \beta$, and using in the $\norm{\ket{\eta_k}} = \beta^k$ case the fact that $1-\alpha^2 = \beta^2$.

Finally, writing $\eps = \eps(n)$ it holds that
\begin{equation*}
	\beta^T
	= \beta^{2^t}
	= \beta^{2^{\ceil{\log \log(1/\eps)} + 7}}
	\le \beta^{128 \log(1/\eps)}
	= \eps^{128 \log(1/\beta)}
	\le \eps^{8.36}
	\le \eps \cdot (1/2)^{7.36}
	\le 0.01 \eps,
\end{equation*}
so
\begin{equation*}
	\Norm{A_n^f \ket\zs - \Paren{\gamma \ket{0^t} \ket\psi + \sqrt{1 - \gamma^2} \ket\tau} \ket{z}}
	\le 1.57 \cdot 2.86 \cdot 2 \beta^T
	\le \eps.
\end{equation*}

\subsection{Computing $f$ in $\poly(n)$ space} \label{sec:qps}
Recall from \cref{alg:oss} that the oracle $f$ encodes, for each $0 \le j < T$, a description of the Clifford unitary $C_j$ and the values $\sr(\bra{\eta_j} C_j \ket{x})$ for $x \in \cube n$. The problem is that $\sr(\bra{\eta_j} C_j \ket{x})$ depends discontinuously on $\bra{\eta_j} C_j \ket{x}$, and $\bra{\eta_j} C_j \ket{x}$ can only be computed approximately due to (exponentially small) error in the description of $\ket{\eta_j}$ and in floating-point arithmetic. Therefore we will use a slightly different oracle $f^\prime$. Let $\delta = 0.01 \cdot \beta^{2T} \ge \exp(-\poly(n))$; we will use $\delta$ to define bounds on the floating-point error in certain calculations.
\paragraph*{The new oracle $f^\prime$}
For a vector $\ket\eta \in \Paren{\C^2}^{\otimes n}$ and a Clifford unitary $C$, let
\begin{equation*}
	\ket{p_{\eta, C}^\prime} = C \cdot 2^{-n/2} \sum_{\mathclap{x \in \cube n}} \sr\Paren{\wt{\bra\eta C \ket x}} \ket x,
\end{equation*}
where $\wt{\bra\eta C \ket x}$ is a value computable in $\poly(n)$ space (given descriptions of $\ket\eta$ and $C$) such that $\Mag{\wt{\bra\eta C \ket x} - \bra\eta C \ket x} \le 2^{-n/2} \delta$. For example we may compute $\wt{\bra\eta C \ket x}$ by a ``sum over histories" argument, i.e.\ write $C = R_1 \dotsb R_{11}$ as the product of the rounds $R_i$ comprising the description of $C$, and use that
\begin{equation*}
	\bra{\eta} C \ket{x} = \sum_{\mathclap{y_0, \dotsc, y_{10} \in \cube n}} \ip{\eta}{y_0} \cdot  \prod_{i=1}^{10} \bra{y_{i-1}} R_i \ket{y_i} \cdot \bra{y_{10}} R_{11} \ket{x}.
\end{equation*}

For a state $\ket\eta$ and Clifford unitary $C$, if $|\mrm{Re}(\bra\eta C \ket x)| > 2^{-n/2} \delta$ then $\sr\Paren{\wt{\bra\eta C \ket x}} = \sr\Paren{\bra\eta C \ket x}$, so by the triangle inequality
\begin{align*}
	\Mag{\mrm{Re} \Paren{\ip \eta {p_{\eta, C}^\prime}} - \mrm{Re} \Paren{\ip \eta {p_{\eta, C}}}}
	&= \Mag{2^{-n/2} \sum_{\mathclap{x \in \cube n}} \Paren{\sr\Paren{\wt{\bra\eta C \ket x}} - \sr\Paren{\bra\eta C \ket x}} \mrm{Re}(\bra\eta C \ket x)} \\
	&\le 2^{-n/2} \sum_{\mathclap{x \in \cube n}} 2 \cdot 2^{-n/2} \delta
	= 2\delta.
\end{align*}
Therefore by \cref{lem:a5}, for all states $\ket\eta$ there exists a Clifford unitary $C$ such that $\mrm{Re}\Paren{\ip \eta {p^\prime_{\eta, C}}} \ge \alpha - 2 \delta$.

For $k \ge 0$, given states $\ket{\phi_0^\prime}, \dotsc, \ket{\phi_{k-1}^\prime}$, let $\ket{\eta_k^\prime} = \ket\psi - \alpha \sum_{j=0}^{k-1} \beta^j \ket{\phi_j^\prime}$, and let $C_k^\prime$ be a Clifford unitary such that the state $\ket{\phi_k^\prime} = \ket{p_{\eta_k^\prime, C_k^\prime}^\prime}$ satisfies $\mrm{Re} \Paren{\ip {\eta_k^\prime} {\phi_k^\prime}} \ge \Paren{\alpha - 2 \delta} \Norm{\ket{\eta_k^\prime}} - \delta$. Let $f^\prime$ be the oracle that encodes, for each $0 \le j < T$, the description of $C_j^\prime$ and $\sr\Paren{\wt{\bra{\eta_j^\prime} C_j^\prime \ket{x}}}$ for $x \in \cube n$.

\paragraph*{Computing $f^\prime$ in $\poly(n)$ space}
Given the description of $\ket{\eta^\prime_k}$, a valid Clifford unitary $C_k^\prime$ can be found in $\poly(n)$ space by performing a brute-force search for a Clifford unitary $C$ such that $\mrm{Re}\Paren{\ip {\eta^\prime_k} {p^\prime_{\eta^\prime_k, C}}} \ge \Paren{\alpha - 2 \delta} \Norm{\ket{\eta_k^\prime}}$. (The ``extra" $\delta$ term in the definition of $C_k^\prime$ allows for floating-point error in the calculation of $\mrm{Re}\Paren{\ip {\eta^\prime_k} {p^\prime_{\eta^\prime_k, C}}} - \Paren{\alpha - 2 \delta} \Norm{\ket{\eta_k^\prime}}$ during this search.) The description of the vector $\ket{\eta_{k+1}^\prime} = \ket{\eta_k^\prime} - \alpha \beta^k \ket{\phi_k^\prime}$ can be subsequently computed in $\poly(n)$ space. Since $\ket{\eta_0^\prime} = \ket\psi$, it follows by induction that descriptions of $C_k^\prime$ and $\ket{\eta_{k+1}^\prime}$ for $k \ge 0$ can be computed in $(k+1) \poly(n)$ space, by answering queries to individual bits of the description of $\ket{\eta_k^\prime}$ recursively. Since $T \le \poly(n)$, descriptions of $C_k^\prime$ and $\ket{\eta_k^\prime}$ for $0 \le k < T$ can be computed in $\poly(n)$ space. Finally, the value $\sr\Paren{\wt{\bra{\eta_k^\prime} C_k^\prime \ket x}}$ can by definition be computed in $\poly(n)$ space given descriptions of $C_k^\prime$ and $\ket{\eta_k^\prime}$.

\paragraph*{Constructing $\ket\psi$ using $f^\prime$}
We now show that \cref{eq:help} still holds with $f^\prime$ substituted for $f$. By reasoning similar to that in \cref{sec:hpc}, there exist a state $\ket{\tau^\prime}$ and a string $z^\prime$ such that $\Paren{\bra{0^t} \otimes I} \ket{\tau^\prime} = 0$ and
\begin{equation*}
	\Norm{A_n^{f^\prime} \ket\zs - \Paren{\gamma \ket{0^t} \ket\psi + \sqrt{1 - \gamma^2} \ket{\tau^\prime}} \ket{z^\prime}}
	\le 1.57 \cdot 2.86 \Paren{\beta^T + \Norm{\ket{\eta_T^\prime}}}.
\end{equation*}

We prove by induction on $k$ that $\Norm{\ket{\eta_k^\prime}}^2 \le \beta^{2k} + 0.1 \beta^{2T} \sum_{j = 0}^{k-1} \beta^j$ for all $k$. The case $k=0$ holds trivially. If the claim holds for $k$, then (similarly to in \cref{sec:hpc})
\begin{align*}
	\Norm{\ket{\eta_{k+1}^\prime}}^2
	&\le \Norm{\ket{\eta_k^\prime}}^2 - 2 \alpha \beta^k \Paren{\Paren{\alpha - 2 \delta} \Norm{\ket{\eta_k^\prime}} - \delta} + \alpha^2 \beta^{2k} \\
	&= \Norm{\ket{\eta_k^\prime}}^2 - 2 \alpha^2 \beta^k \Norm{\ket{\eta_k^\prime}} + \alpha^2 \beta^{2k} + \beta^k \cdot 2 \alpha \Paren{2 \Norm{\ket{\eta_k^\prime}} + 1} \delta.
\end{align*}
By the triangle inequality
\begin{equation*}
	\Norm{\ket{\eta_k^\prime}}
	= \Norm{\ket\psi  - \alpha \sum_{j=0}^{k-1} \beta^j \ket{\phi_j^\prime}}
	\le 1 + \alpha \sum_{j=0}^{k-1} \beta^j
	\le 1 + \frac\alpha{1-\beta}
	= 1 + \frac1\gamma,
\end{equation*}
so recalling that $\delta = 0.01 \beta^{2T}$ it holds that
\begin{equation*}
	2\alpha \Paren{2 \Norm{\ket{\eta_k^\prime}} + 1} \delta
	\le 2\alpha \Paren{2 \Paren{1 + \frac1\gamma} + 1} \cdot 0.01 \beta^{2T}
	< 0.1 \beta^{2T},
\end{equation*}
and therefore
\begin{equation*}
	\Norm{\ket{\eta_{k+1}^\prime}}^2
	\le \Norm{\ket{\eta_k^\prime}}^2 - 2 \alpha^2 \beta^k \Norm{\ket{\eta_k^\prime}} + \alpha^2 \beta^{2k} + \beta^k \cdot 0.1 \beta^{2T}.
\end{equation*}
The rest of the inductive argument follows by reasoning similar to that in \cref{sec:hpc}.

Therefore $\Norm{\ket{\eta_T^\prime}} \le \sqrt{\beta^{2T} + 0.1 \beta^{2T} / (1-\beta)} < 1.7 \beta^T$, and the rest of the proof is similar to that in \cref{sec:hpc}.

\section{State synthesis algorithms: removing the postselection} \label{sec:msa}
In this section we formally state and prove both the non-clean, one-query version and the clean, four-query version of \cref{thm:mi}. We also give a clean, ten-query state synthesis algorithm that has two advantages compared to the four-query algorithm. First, the ten-query algorithm is simpler. Second, the oracle in the ten-query algorithm requires fewer input bits, which will be relevant when we prove circuit upper bounds for approximately constructing arbitrary states (i.e.\ \cref{thm:ubs}).

All of these algorithms invoke the algorithm from \cref{lem:help}. Let $\gamma$ be the constant from \cref{lem:help}. Given an $n$-qubit state $\ket\psi$ and parameter $\eps$, when we say ``define $A, f, \ket\tau, z, t$ as in \cref{lem:help} with respect to $\ket\psi$ and error tolerance $\eps$", we mean that $A$ is the circuit $A_n$ from \cref{lem:help} and all other variables have the same meaning as in \cref{lem:help}. We will write $\eps = \eps(n)$ and $t = t(n)$ when $n$ is implicit. In the four- and ten-query algorithms not all of the queries will be to precisely the same function, but this can easily be addressed as discussed in the paragraph after \cref{que:ssp}. Although our results will be stated in terms of \qaczf circuits, similar results hold for circuits consisting of one- and two-qubit gates by \cref{lem:qsim}.
\subsection{The one-query algorithm} \label{sec:oqa}

We prove the following by using parallel repetition to boost the success probability from \cref{lem:help}:

\begin{thm} \label{thm:main}
	Let $\eps$ be a function such that $\eps(n) \ge \exp(-\poly(n))$ and $\eps(n)$ is computable in $\poly(n)$ time for all $n$. Then there is a uniform sequence of $\poly(n)$-qubit \qaczf circuits $(C_n)_n$, each making one query to a classical oracle, such that for every $n$-qubit state $\ket\psi$ there exists a classical oracle $f = f_{\ket\psi}$ such that the reduced state $\rho$ on the first $n$ qubits of $C_n^f \ket\zs$ satisfies $\td(\rho, \psi) \le \eps(n)$. Furthermore there is an algorithm that takes as input the description of an $n$-qubit state $\ket\psi$ and a string $x$, runs in $\poly(n)$ space, and outputs $f_{\ket\psi} (x)$.
\end{thm}

\begin{proof}
	Let $s = \ceil*{2 \ln(2/\eps) / \gamma^2} \le \poly(n)$, and define $A, f, \ket\tau, z, t$ as in \cref{lem:help} with respect to $\ket\psi$ and error tolerance $\eps / (2 s) \ge \exp(-\poly(n))$. The algorithm is presented in \cref{alg:oqss}, where $\reg A_k$ is a $t$-qubit register, $\reg B_k$ is an $n$-qubit register, and $\reg C_k$ is a $|z|$-qubit register for all $k \in [s]$. \cref{alg:oqss} is phrased in terms of multiple parallel queries, but these can be merged into a single query using \cref{eq:par-quer}. The \qaczf implementation of \cref{oqline:5} uses \cref{lem:qaczf-swap}.
	
	\begin{algorithm}
		\caption{One-query state synthesis}
		\label{alg:oqss}
		\begin{algorithmic}[1]
			\For{$k \in [s]$ in parallel}
			\State Apply $A^f$ in $\reg A_k \reg B_k \reg C_k$. \Comment Merge queries using \cref{eq:par-quer}.
			\EndFor \hlabel{oqline:3}
			\Ctrl{the classical state $\ket{x_1}_{\reg A_1} \dotsb \ket{x_s}_{\reg A_s}$} \hlabel{oqline:4}
			\If{there exists $k$ such that $x_k = 0^t$} \Return $\reg B_k$ for the smallest such $k$. \hlabel{oqline:5}
			\Else{} \Return an arbitrary $n$-qubit state. \hlabel{oqline:6}
			\EndIf \hlabel{oqline:7}
			\EndCtrl \hlabel{oqline:8}
		\end{algorithmic}
	\end{algorithm}
	
	Let
	\begin{equation*}
		\ket{\tilde \varphi} = \bigotimes_{k=1}^s \Paren{\gamma \ket{0^t}_{\reg A_k} \ket\psi_{\reg B_k} + \sqrt{1 - \gamma^2} \ket\tau_{\reg A_k \reg B_k}} \ket{z}_{\reg C_k},
	\end{equation*}
	and let $\tilde \rho$ denote the $n$-qubit output state produced by running \cref{oqline:4,oqline:5,oqline:6,oqline:7,oqline:8} on $\ket{\tilde \varphi}$. If the $\reg A_k$ registers of $\ket{\tilde \varphi}$ are measured in the standard basis, then the probability that none of the measurement outcomes are $0^t$ is $\Paren{1 - \gamma^2}^s$, so by \cref{eq:fvdg}
	\begin{equation*}
		\td\Paren{\psi, \tilde\rho} \le \Paren{1 - \gamma^2}^{s/2} \le \exp\Paren{-\gamma^2 s / 2} \le \eps/2.
	\end{equation*}
	Let $\ket\varphi$ denote the state of the system after \cref{oqline:3}. Then by \cref{eq:tdc,eq:tdps} and the triangle inequality
	\begin{equation*}
		\td\Paren{\tilde\rho, \rho} \le
		\td\Paren{\tilde\varphi, \varphi} \le
		\Norm{\ket{\tilde\varphi} - \ket\varphi} \le
		s \cdot \eps/(2s) =
		\eps/2,
	\end{equation*}
	so by the triangle inequality
	\begin{equation*}
		\td\Paren{\psi, \rho}
		\le \td\Paren{\psi, \tilde\rho} + \td\Paren{\tilde\rho, \rho}
		\le \eps/2 + \eps/2
		= \eps. \qedhere
	\end{equation*}
\end{proof}

\subsection{The ten-query algorithm} \label{sec:eqa}

We prove the following by using amplitude amplification to boost the success probability from \cref{lem:help}:

\begin{thm} \label{thm:tqa}
	Let $\eps$ be a function such that $\eps(n) \ge \exp(-\poly(n))$ and $\eps(n)$ is computable in $\poly(n)$ time for all $n$. Then there is a uniform sequence of $\poly(n)$-qubit \qaczf circuits $(C_n)_n$, each making ten queries to a classical oracle, such that for every $n$-qubit state $\ket\psi$ there exists a classical oracle $f = f_{\ket\psi}$ such that $\Norm{C_n^f \ket\zs - \ket\psi \ket\zs} \le \eps(n)$. Furthermore there is an algorithm that takes as input the description of an $n$-qubit state $\ket\psi$ and a string $x$, runs in $\poly(n)$ space, and outputs $f_{\ket\psi} (x)$.
\end{thm}

\begin{proof}
	Define $A, f, \ket\tau, z, t$ as in \cref{lem:help} with respect to $\ket\psi$ and error tolerance $\eps/(9 \sqrt 2)$. Since $\sin(\pi/18) < 0.174 < 0.18 < \gamma$, there exists a one-qubit gate $G$ such that
	\begin{equation*}
		G \ket0 = \frac{\sin(\pi/18)}\gamma \ket0 + \sqrt{1 - \Paren{\frac{\sin(\pi/18)} \gamma}^2} \ket1.
	\end{equation*}
	Let $\ket{\theta} = \Paren{G \otimes A^f} \ket\zs$. The algorithm is described in \cref{alg:eqss}.\footnote{Inspection of the proof of \cref{lem:help} reveals that the last query (i.e.\ uncomputing $z$) can be computed in $\poly(n)$ space, as required by the theorem. The idea to use $G$ to artificially decrease the initial ``success" amplitude was suggested to us by Wiebe~\cite{Wie21a}.}
	
	\begin{algorithm}
		\caption{Ten-query state synthesis}
		\label{alg:eqss}
		\begin{algorithmic}[1]
			\State Construct $\Paren{\Paren{2\theta - I} \cdot \Paren{\Paren{I - 2\kb{0^{1+t}}} \otimes I}}^4 \ket{\theta}$.
			\State Use one more query to uncompute $z$.
		\end{algorithmic}
	\end{algorithm}
	
	By \cref{lem:err} it suffices to prove that if we substitute the state
	\begin{equation*}
		\ket{\tilde\theta} = G\ket0 \otimes \Paren{\gamma \ket{0^t} \ket\psi + \sqrt{1 - \gamma^2} \ket\tau} \ket{z}
	\end{equation*}
	for each occurrence of $\ket\theta$ in \cref{alg:eqss}, then the output state is exactly $\ket\psi \ket\zs$. Since $\Paren{\bra{0^{1+t}} \otimes I} \ket{\tilde \theta} = \sin(\pi/18) \ket\psi \ket{z}$, we may write
	\begin{equation*}
		\ket{\tilde\theta} = \Paren{\sin(\pi/18) \ket{0^{1+t}} \ket\psi + \cos(\pi/18) \ket\varphi} \ket{z}
	\end{equation*}
	for some state $\ket\varphi$ such that $\Paren{\bra{0^{1+t}} \otimes I} \ket\varphi = 0$. By well-known arguments (cf.\ the proof of correctness of Grover's algorithm~\cite{NC10}) it follows that if $\ket{\tilde\theta}$ is substituted for $\ket\theta$, then the output state is
	\begin{equation*}
		\Paren{\sin(9 \cdot \pi/18) \ket{0^{1+t}} \ket\psi + \cos(9 \cdot \pi/18) \ket\varphi} \ket\zs
		= \ket\psi \ket\zs. \qedhere
	\end{equation*}
\end{proof}

\subsection{The four-query algorithm} \label{sec:fqa}

The following statement is identical to \cref{thm:tqa} except with ``four" instead of ``ten", and is proved by a combination of the ideas from \cref{sec:oqa,sec:eqa}:

\begin{thm} \label{thm:fqa}
	Let $\eps$ be a function such that $\eps(n) \ge \exp(-\poly(n))$ and $\eps(n)$ is computable in $\poly(n)$ time for all $n$. Then there is a uniform sequence of $\poly(n)$-qubit \qaczf circuits $(C_n)_n$, each making four queries to a classical oracle, such that for every $n$-qubit state $\ket\psi$ there exists a classical oracle $f = f_{\ket\psi}$ such that $\Norm{C_n^f \ket\zs - \ket\psi \ket\zs} \le \eps(n)$. Furthermore there is an algorithm that takes as input the description of an $n$-qubit state $\ket\psi$ and a string $x$, runs in $\poly(n)$ space, and outputs $f_{\ket\psi} (x)$.
\end{thm}

\begin{proof}
	Let $\delta = \sqrt{1 - \gamma^2}$. Let $s$ be the smallest power of 2 that is at least $\log(4/\eps) / \log(1/\delta)$, and observe that $s \le 2 \log(4/\eps) / \log(1/\delta) \le \poly(n)$. Define $A, f, \ket\tau, z, t$ as in \cref{lem:help} with respect to $\ket\psi$ and error tolerance $\eps/(\sqrt 2 \cdot 8s)$. 
	
	First we show how to approximately construct $\ket\tau \ket{z}$ with three queries using amplitude amplification. Since $\sin(\pi/6) = 1/2 < 0.98 \approx \delta$, there exists a one-qubit gate $G$ such that
	\begin{equation*}
		G\ket0 = \frac {\sin(\pi/6)} \delta \ket0 + \sqrt {1 - \Paren{\frac {\sin(\pi/6)} \delta}^2} \ket1.
	\end{equation*}
	Let $\ket{\theta} = \Paren{G \otimes A^f} \ket\zs$ and
	\begin{equation*}
		U^f = \Paren{2\theta - I} \Paren{\Paren{I - 2\kb0 \otimes \Paren{I - \kb{0^t}}} \otimes I} \Paren{G \otimes A^f}.
	\end{equation*}
	Then $U^f$ can be efficiently implemented with three queries, and if $A^f$ \emph{exactly} constructs $\Paren{\gamma \ket{0^t} \ket\psi + \delta \ket\tau} \ket{z}$ then $U^f$ \emph{exactly} constructs $\ket0 \ket\tau \ket{z}$ by reasoning similar to that in \cref{sec:eqa}.
	
	Since
	\begin{equation*}
		\sum_{k=0}^{s-1} \delta^k \ket k
		= \Paren{\ket0 + \delta^{s/2} \ket1} \otimes \Paren{\ket0 + \delta^{s/4} \ket1} \otimes \dotsb \otimes \Paren{\ket0 + \delta \ket1},
	\end{equation*}
	there exists a tensor product $L$ of one-qubit gates such that
	\begin{equation*}
		L\ket{0^{\log s}} = \frac \gamma {\sqrt{1 - \delta^{2s}}} \sum_{k=0}^{s-1} \delta^k \ket k.
	\end{equation*}
	The algorithm is presented in \cref{alg:fqss}, where $\reg A_k$ is a $t$-qubit register, $\reg B_k$ is an $n$-qubit register, and $\reg C_k$ is a $|z|$-qubit register for all $0 \le k < s$; additionally $\reg K$ is a $(\log s)$-qubit register and $\reg O$ is an $n$-qubit register. The extra ancilla qubit in \cref{line:ab} accounts for the fact that $U^f$ acts on one more qubit than $A^f$ does. The \qaczf implementation of \cref{line:foo} uses \cref{lem:qaczf-swap}, and the \qaczf implementations of \cref{line:ab,line:bar} use \cref{lem:cqacz}.\footnote{Although not directly implied by \cref{lem:cqacz}, inspection of the proof of \cref{lem:cqacz} reveals that if $(B_n)_n$ is a \emph{uniform} sequence of polynomial-size \qaczf circuits then $(\crl B_n)_n$ can be implemented by a \emph{uniform} sequence of polynomial-size \qaczf circuits.}
	
	\begin{algorithm}
		\caption{Four-query state synthesis}
		\label{alg:fqss}
		\begin{algorithmic}[1]
			\For{$0 \le k < s$ in parallel}
				\State Apply $A^f$ in $\reg A_k \reg B_k \reg C_k$. \Comment Merge queries using \cref{eq:par-quer}.
			\EndFor
			\Ctrl{the classical state $\ket{x_0}_{\reg A_0} \dotsb \ket{x_{s-1}}_{\reg A_{s-1}}$}
				\If{there exists $k$ such that $x_k = 0^t$} $\reg K \gets$ the smallest such $k$.
				\EndIf
			\EndCtrl
			\Ctrl{the classical state $\ket k_{\reg K}$}
				\State Swap $\reg B_k$ and $\reg O$. \hlabel{line:foo}
				\State Uncompute $\ket{z}_{\reg C_k}$, controlled on $\ket{z}_{\reg C_j}$ for some $j \neq k$.
				\For{$0 \le j < s$ in parallel} \Comment Merge queries using \cref{eq:par-quer}.
					\If{$j < k$} apply $\adj{\Paren{U^f}}$ in $\reg A_j \reg B_j \reg C_j$ (with one extra ancilla qubit). \hlabel{line:ab}
					\ElsIf{$j > k$} apply $\adj{\Paren{A^f}}$ in $\reg A_j \reg B_j \reg C_j$. \hlabel{line:bar}
					\EndIf
				\EndFor
			\EndCtrl
			\State Apply $\adj L$ in $\reg K$.
		\end{algorithmic}
	\end{algorithm}
	
	Assume for now that $A^f$ \emph{exactly} constructs $\Paren{\gamma \ket{0^t} \ket\psi + \delta \ket\tau} \ket{z}$. Let $\ket{\Psi_\ell}$ denote the state of the system after line $\ell$, up to omitting registers in the all-zeros state for brevity. Then
	\begin{equation*}
		\ket{\Psi_3}
		= \bigotimes_{k=0}^{s-1} A^f \ket\zs_{\reg A_k \reg B_k \reg C_k}
		= \bigotimes_{k=0}^{s-1} \Paren{\gamma \ket{0^t}_{\reg A_k} \ket\psi_{\reg B_k} + \delta \ket \tau_{\reg A_k \reg B_k}} \ket{z}_{\reg C_k},
	\end{equation*}
	so
	\begin{equation*}
		\Norm{\ket{\Psi_7} -
			\sum_{k=0}^{s-1} \delta^k \gamma \ket{k}_{\reg K} \otimes
			\bigotimes_{j=0}^{k-1} \ket\tau_{\reg A_j \reg B_j} \ket{z}_{\reg C_j}
			\otimes \ket{0^t}_{\reg A_k} \ket\psi_{\reg B_k} \ket{z}_{\reg C_k}
			\otimes \bigotimes_{\mathclap{j=k+1}}^{s-1} A^f \ket\zs_{\reg A_j \reg B_j \reg C_j}}
		\le \delta^s,
	\end{equation*}
	so
	\begin{equation*}
		\Norm{\ket{\Psi_{16}} -
			\ket\psi_{\reg O} \otimes
			\sum_{k=0}^{s-1} \delta^k \gamma \ket{k}_{\reg K}
			\otimes \ket\zs_{\reg A_0 \reg B_0 \reg C_0 \dotsb \reg A_{s-1} \reg B_{s-1} \reg C_{s-1}}}
		\le \delta^s,
	\end{equation*}
	so
	\begin{equation*}
		\Norm{\ket{\Psi_{17}} - \sqrt{1 - \delta^{2s}} \ket\psi \ket\zs} \le \delta^s.
	\end{equation*}
	By the triangle inequality it follows that
	\begin{equation*}
		\Norm{\ket{\Psi_{17}} - \ket\psi \ket\zs}
		\le 1 - \sqrt{1 - \delta^{2s}} + \delta^s
		\le \delta^{2s} + \delta^s
		\le 2\delta^s.
	\end{equation*}
	
	Now remove the assumption that $A^f$ constructs $\Paren{\gamma \ket{0^t} \ket\psi + \delta \ket\tau} \ket{z}$ exactly. \cref{alg:fqss} makes $4s$ queries to $\crl{A^f}$ and its inverse, so by \cref{lem:err} it follows that the \emph{actual} output state $\ket{\Psi_{17}}$ satisfies
	\begin{equation*}
		\Norm{\ket{\Psi_{17}} - \ket\psi \ket\zs}
		\le 2\delta^s + \sqrt 2 \cdot 4s \cdot \eps/(\sqrt 2 \cdot 8s)
		\le \eps/2 + \eps/2
		= \eps. \qedhere
	\end{equation*}
\end{proof}

\section{State complexity classes} \label{sec:scc}
\label{sec:state_classes}
In this section we define various state complexity classes and establish some basic facts about them as preparation for the proof that $\cc{statePSPACE} \subseteq \cc{stateQIP}(6)$. Although for simplicity these classes are defined in terms of state sequences where the $n$'th state is on $n$ qubits, the definitions (and related results) generalize easily to sequences where the $n$'th state is on $\poly(n)$ qubits. Much of the language in this section is closely modeled on passages from Rosenthal and Yuen~\cite{RY21} and Metger and Yuen~\cite{MY23}.
\subsection{$\cc{polyL}$-explicit state sequences} \label{sec:esf}

Recall from \cref{sec:prelim} that we define an $\eps$-precision description of a pure state $\sum_{x \in \cube n} \alpha_x \ket x$ to be a tuple $\Paren{\tilde\alpha_x}_{x \in \cube n}$ of complex numbers specified exactly in binary such that $\Mag{\tilde\alpha_x - \alpha_x} \le \eps$ for all $x$. We define a similar notion for mixed states: an \emph{$\eps$-precision description} of a mixed state $\sum_{x,y \in \cube n} \rho_{x,y} \ketbra{x}{y}$ is a tuple $\Paren{\tilde \rho_{x,y}}_{x,y \in \cube n}$ of complex numbers specified exactly in binary such that $\Mag{\tilde\rho_{x,y} - \rho_{x,y}} \le \eps$ for all $x,y$.

\begin{dfn}[$\cc{polyL}$-explicit state sequences] \label{def:epsf}
	Let $\ket{\psi_n}$ be an $n$-qubit pure state for all $n$. We call the sequence $(\ket{\psi_n})_n$ \emph{$\cc{polyL}$-explicit} if for all functions of the form $\eps(n) = \exp(-\poly(n))$, there is an algorithm that on input $n$ outputs an $\eps(n)$-precision description of $\ket{\psi_n}$ using space $\poly(n)$ (i.e.\ space polylogarithmic in the output length).
	
	Similarly, let $\rho_n$ be an $n$-qubit mixed state for all $n$. We call the sequence $(\rho_n)_n$ \emph{$\cc{polyL}$-explicit} if for all functions of the form $\eps(n) = \exp(-\poly(n))$, there is an algorithm that on input $n$ outputs an $\eps(n)$-precision description of $\rho_n$ using space $\poly(n)$.
\end{dfn}

\begin{lem} \label{lem:exp}
	Let $(\rho_n)_n$ be a $\cc{polyL}$-explicit sequence of rank-$1$ mixed states. Then there is a $\cc{polyL}$-explicit sequence of pure states $(\ket{\psi_n})_n$ such that $\rho_n = \kb{\psi_n}$ for all $n$.
\end{lem}

\begin{proof}
	Fix $n$ and write $\rho = \rho_n = \sum_{x,y \in \cube n} \rho_{x,y} \ketbra{x}{y}$. Let $\Paren{\tilde\rho_{x,y}}_{x,y \in \cube n}$ be a $\Paren{\frac14 \cdot 2^{-n}}$-precision description of $\rho$ computable in $\poly(n)$ space. Since $\tr\rho = 1$ there exists a string $x$ such that $\rho_{x,x} \ge 2^{-n}$, implying that $\tilde\rho_{x,x} \ge \rho_{x,x} - \frac14 \cdot 2^{-n} \ge \frac34 \cdot 2^{-n}$. Let $y$ be the lexicographically first string such that $\tilde\rho_{y,y} \ge \frac34 \cdot 2^{-n}$ (which we have just shown to exist) and observe that $\rho_{y,y} \ge \tilde\rho_{y,y} - \frac14 \cdot 2^{-n} \ge \frac12 \cdot 2^{-n}$. Let
	\begin{equation*}
		\ket\psi
		= \ket{\psi_n}
		= \frac{\rho \ket y} {\sqrt{\rho_{y,y}}}
		= \sum_{x \in \cube n} \frac {\rho_{x,y}} {\sqrt{\rho_{y,y}}} \ket x.
	\end{equation*}
	Since $\rho$ is rank-1 it is easy to see that $\rho = \psi$.

	For $\eps = \exp(-\poly(n))$ an $\eps$-precision description of $\ket\psi$ can be computed in $\poly(n)$ space as follows. Let $\delta = \frac1{64} \cdot 2^{-2n} \eps^2 \ge \exp(-\poly(n))$ and let $(\sigma_{x,y^\prime})_{x,y^\prime \in \cube n}$ be a $\delta$-precision description of $\rho$ computable in $\poly(n)$ space. First compute $y$ (using that $\tilde\rho$ can be computed in $\poly(n)$ space), and then output $\Paren{\sigma_{x,y} / \sqrt{\sigma_{y,y}}}_{x \in \cube n}$.
	
	This algorithm is correct, because by the triangle inequality
	\begin{align*}
		\Mag{\frac{\sigma_{x,y}} {\sqrt{\sigma_{y,y}}} - \frac{\rho_{x,y}} {\sqrt{\rho_{y,y}}}}
		&= \Mag{\frac {\sigma_{x,y} \sqrt{\rho_{y,y}} - \sqrt{\sigma_{y,y}} \rho_{x,y}} {\sqrt{\sigma_{y,y} \rho_{y,y}}}}
		\le \frac{\sqrt{\rho_{y,y}} \cdot \Mag{\sigma_{x,y} - \rho_{x,y}} + \Mag{\rho_{x,y}} \cdot \Mag{\sqrt{\rho_{y,y}} - \sqrt{\sigma_{y,y}}}} {\sqrt{\Paren{\rho_{y,y} - \delta} \rho_{y,y}}} \\
		&\le \frac {\delta + \sqrt{\Mag{\rho_{y,y} - \sigma_{y,y}}}} {\sqrt{\Paren{\frac12 \cdot 2^{-n} - \delta} \cdot \frac12 \cdot 2^{-n}}}
		\le \frac {2 \sqrt\delta} {\sqrt{\frac18 \cdot 2^{-2n}}}
		\le \eps,
	\end{align*}
	where the second-to-last inequality uses that $\delta \le \frac14 \cdot 2^{-n}$.
\end{proof}

\subsection{The class $\cc{statePSPACE}$}

For convenience we use the universal gate set $\{ H, \mathit{CNOT}, T \}$~\cite{NC10} in the following definition, although our results hold for any universal gate set consisting of gates with algebraic entries.

\begin{dfn}[General quantum circuits and space-uniformity] \label{def:gqc}
	A \emph{general quantum circuit} is a circuit consisting of gates from the set $\{H, \mathit{CNOT}, T\}$ as well as non-unitary gates that (a) introduce new qubits initialized in the zero state, (b) trace them out, or (c) measure them in the standard basis. A general quantum circuit uses space $s$ if at most $s$ qubits are involved at any time step of the computation. The description of a general quantum circuit is the sequence of its gates (unitary or non-unitary) along with a specification of which qubits they act on.
	
	We call a sequence $(C_n)_n$ of general quantum circuits \emph{space-uniform} if $C_n$ uses space $\poly(n)$, and there is an algorithm that on input $n$ uses space $\poly(n)$ and outputs the (possibly exponentially long) description of $C_n$.
\end{dfn}

\begin{dfn}[$\statePSPACE$ and variants thereof]
	For $\delta: \N \to [0,\infty)$, let $\statePSPACE_\delta$ be the class of all sequences of mixed states $(\rho_n)_n$ such that each $\rho_n$ is a state on $n$ qubits, and there exists a space-uniform sequence of general quantum circuits $(C_n)_n$ such that for all sufficiently large $n$, the circuit $C_n$ takes no inputs and $C_n$ outputs a mixed state $\sigma_n$ such that $\td(\sigma_n, \rho_n) \le \delta(n)$. Let $\statePSPACE = \bigcap_p \statePSPACE_{1/p}$ and $\sPe = \bigcap_p \statePSPACE_{\exp(-p)}$ where $p$ ranges over all polynomials.
\end{dfn}

We abuse notation and write $(\ket{\psi_n})_n \in \statePSPACE_\delta$ if $(\kb{\psi_n})_n$ is in $\statePSPACE_\delta$. Also recall that the definitions of state complexity classes such as $\statePSPACE_\delta$ generalize easily to sequences where the $n$'th state is on $\poly(n)$ qubits. With this in mind we can state the following result, which in particular implies that $\sPe$ is closed under purification:

\begin{lem}[{\cite[part of Theorem 6.1]{MY23}\protect\footnote{As of this writing the $\eps$ term is omitted from \cite[Theorem 6.1]{MY23}, but inspection of their proof reveals that this omission is an error.}}] \label{lem:pur}
	Let $(\rho_n)_n \in \cc{statePSPACE}_\delta$ be a sequence of mixed states for some function $\delta$. Then there exists a sequence of pure states $(\ket{\psi_n})_n \in \bigcap_{\eps(n) = \exp(-\poly(n))} \cc{statePSPACE}_{2\sqrt\delta + \eps}$ such that $\ket{\psi_n}$ is a purification of $\rho_n$ for all $n$.
\end{lem}

We also use the following:

\begin{lem} \label{lem:msp}
	Every sequence of mixed states in $\sPe$ is $\cc{polyL}$-explicit.
\end{lem}
\begin{proof}
	Metger and Yuen~\cite[Lemma 6.2]{MY23} proved that every sequence of mixed states in $\cc{statePSPACE}_0$ is $\cc{polyL}$-explicit. The general case follows by the triangle inequality.
\end{proof}

\begin{rmk*}
	The high-level idea behind the proof of \cite[Lemma 6.2]{MY23} is that tomography of states in $\cc{statePSPACE}_0$ can be done in $\cc{BQPSPACE}$, and $\cc{BQPSPACE} = \cc{PSPACE}$~\cite{watrous03complexity}. The proof of $\cc{BQPSPACE} = \cc{PSPACE}$ relies on the assumption that the gates used in \cref{def:gqc} have algebraic entries, which is why we imposed this requirement.
\end{rmk*}

\subsection{Quantum interactive protocols} \label{sec:gqp}

Since in quantum computing the standard model of computation is the quantum circuit model (rather than quantum Turing machines), we model the verifier in a quantum interactive protocol as a sequence of \emph{verifier circuits}, one for each input length. A verifier circuit is itself a tuple of quantum circuits that correspond to the operations performed by the verifier in each round of the protocol. Below we describe this more formally.

The case where the verifier sends the first message is illustrated in \cref{fig:pg}. For a register $\reg A$ let $\linear(\reg A)$ denote the set of density matrices on $\reg A$. A \emph{$2r$-message quantum verifier circuit} $C = (C_j)_{j \in [r+1]}$ is a tuple of general quantum circuits, where $C_1: \linear(\reg W_0) \to \linear(\reg W_1 \reg M_1)$, and $C_j: \linear(\reg W_{j-1} \reg M_{2j-2}) \to \linear(\reg W_j \reg M_{2j-1})$ for $2 \le j \le r$, and $C_{r+1}: \linear(\reg W_r \reg M_{2r}) \to \linear(\reg Z \reg W_{r+1} \reg S)$. A \emph{quantum prover} $P$ for such a verifier circuit $C$ is a tuple of quantum channels $(P_j)_{j \in [r]}$ where $P_j: \linear (\reg Q_{j-1} \reg M_{2j-1}) \to \linear(\reg Q_j \reg M_{2j})$. We think of $\reg W_j$ (resp.\ $\reg Q_j$) as the verifier's (resp.\ prover's) private memory at a given time, and we think of $\reg M_j$ as the $j$'th message. At the end of the protocol, the verifier produces a one-qubit register $\reg Z$ indicating whether to accept or reject, and a register $\reg S$ containing an output state.

\begin{figure*}[t]
	\centering
	\resizebox{0.9\textwidth}{!}{
		\begin{tikzpicture}
			\tikzset{
				block/.style={rectangle,draw,thick,fill=gray!50,inner sep=0pt,minimum width=1.23cm,minimum height=1.25cm},
				point/.style={minimum width=1.25cm},
				every node/.style={scale=1.4,font=\large}
			}
			\coordinate (1) at (-2,3.5){};
			\coordinate (2) at (-2,0){};
			\node[block] (3) at (0,3.5){$C_1$};
			\node[block] (4) at (2.75,0){$P_1$};
			\node[block] (5) at (5.5,3.5){$C_2$};
			\node[block] (6) at (8.25,0){$P_2$};
			\node[point] (7) at (11,3.5){};
			\node[point] (8) at (11,0){};
			\node[point] (9) at (11.75,3.5){};
			\node[point] (10) at (11.75,0){};
			\node[block] (11) at (17.25,3.5){$C_r$};
			\node[block] (12) at (14.5,0){$P_{r-1}$};
			\node[block] (13) at (22.75,3.5){$C_{r+1}$};
			\node[block] (14) at (20,0){$P_r$};
			\coordinate (15) at (24.5,3.5){};
			\coordinate (16) at (24.5,0){};
			
			\draw[-latex] (1) -- (3) node[midway,above]{$\reg W_0$};
			\draw[-latex] (3) -- (5) node[midway,above]{$\reg W_1$};
			\draw[-latex] (5) -- (7) node[midway,above]{$\reg W_2$};
			\draw[-latex] (9) -- (11) node[midway,above]{$\reg W_{r-1}$};
			\draw[-latex] (11) -- (13) node[midway,above]{$\reg W_r$};			
			\draw[-latex] (2) -- (4) node[midway,below]{$\reg Q_0$};
			\draw[-latex] (4) -- (6) node[midway,below]{$\reg Q_1$};
			\draw[-latex] (6) -- (8) node[midway,below]{$\reg Q_2$};
			\draw[-latex] (10) -- (12) node[xshift=10pt,midway,below left]{$\reg Q_{r-2}$};
			\draw[-latex] (12) -- (14) node[midway,below]{$\reg Q_{r-1}$};
			\draw[-latex] (14) -- (16) node[midway,below] {$\reg Q_r$};
			\draw[-latex] ([yshift=-15pt]3.east) to[out=0,in=180] node[midway,right]{$\reg M_1$} ([yshift=15pt]4.west);
			\draw[-latex] ([yshift=15pt]4.east) to[out=0,in=180] node[midway,right]{$\reg M_2$} ([yshift=-15pt]5.west);
			\draw[-latex] ([yshift=-15pt]5.east) to [out=0,in=180] node[midway,right]{$\reg M_3$} ([yshift=15pt]6.west);
			\draw[-latex] ([yshift=15pt]6.east) to [out=0,in=180] node[midway,left]{$\reg M_4$} ([yshift=-15pt]7.west);
			\fill (10.75,1.5) circle (2pt) (11.25,1.5) circle (2pt) (11.75,1.5) circle (2pt);
			\draw[-latex] ([yshift=-15pt]9.east) to[out=0,in=180] node[midway,right]{$\reg M_{2r-3}$} ([yshift=15pt]12.west);
			\draw[-latex] ([yshift=15pt]12.east) to[out=0,in=180] node[midway,right]{$\reg M_{2r-2}$} ([yshift=-15pt]11.west);
			\draw[-latex] ([yshift=-15pt]11.east) to[out=0,in=180] node[midway,right]{$\reg M_{2r-1}$} ([yshift=15pt]14.west);
			\draw[-latex] ([yshift=15pt]14.east) to[out=0,in=180] node[midway,right]{$\reg M_{2r}$} ([yshift=-15pt]13.west);
			\draw[-latex] (13.east) -- (15) node[right]{$\reg W_{r+1}$};
			\draw[-latex] ([yshift=15pt]13.east) to[out=0,in=180] ([yshift=30pt]15) node[right]{$\reg Z$};
			\draw[-latex] ([yshift=-15pt]13.east) to[out=0,in=180] ([yshift=-30pt]15) node[right]{$\reg S$};
		\end{tikzpicture}
	}
	\caption{Generic quantum interactive protocol.}
	\label{fig:pg}
\end{figure*}

Let $x$ denote a string whose length is at most the number of qubits in $\reg W_0$. We write $C(x) \interact P$ to denote the interaction between the verifier circuit $C$ and the prover $P$ on input $x$, which means applying the channels $C_j$ and $P_j$ as pictured in \cref{fig:pg} to the initial state $\ket{x, \zs}_{\reg W_0} \ket\zs_{\reg Q_0}$. We say that $C(x) \interact P$ accepts (resp.\ rejects) if measuring $\reg Z$ in the standard basis yields the outcome $1$ (resp.\ $0$). If $C(x) \interact P$ accepts with nonzero probability, then by the \emph{output of $C(x) \interact P$ conditioned on accepting} we mean the reduced state in $\reg S$ conditioned on $C(x) \interact P$ accepting. In other words if $\rho$ denotes the output of $C_{r+1}$, then the output state conditioned on accepting is
\begin{equation*}
	\Tr_{\reg W_{r+1}} \Paren{\frac {\bra1_{\reg Z} \rho \ket1_{\reg Z}} {\tr{\bra1_{\reg Z} \rho \ket1_{\reg Z}}}}.
\end{equation*}

By dilating we can assume without loss of generality that the prover's channels are all unitary, i.e.\ $P_j(A) = U_j A \adj U_j$ for some unitary $U_j$, and similarly for the verifier. (This is the purpose of the registers $\reg Q_0, \reg Q_r, \reg W_{r+1}$.) We always assume that the prover is unitary, but only sometimes assume that the verifier is unitary.

We can model interactions in which the prover sends the first (nontrivial) message by requiring $\reg M_1$ to only convey the input string $x$ that was in $\reg W_0$. In this case there are only $2r-1$ (nontrivial) messages.

We say that a sequence of quantum verifier circuits $(V_n)_n$ is \emph{uniform} if the total number gates in all circuits in $V_n$ is $\poly(n)$, and the descriptions of the circuits in $V_n$ can be computed in $\poly(n)$ time as a function of $n$. For $m: \N \to \N$, an \emph{$m$-message quantum verifier} is a uniform sequence $(V_n)_n$ of quantum verifier circuits where $V_n$ defines a protocol with $m(n)$ messages. These $m(n)$ messages include messages sent by both the verifier and prover, and do not include the trivial first message sent by the verifier if $m(n)$ is odd.

\subsection{The class $\cc{QIP}(3)$}

The class $\cc{QIP}$ is the standard quantum analogue of the complexity class $\IP$. For our purposes we will only need to define the three-message version of $\cc{QIP}$, known as $\cc{QIP}(3)$. Below we abbreviate $V_{|x|}(x) \interact P$ by $V(x) \interact P$.

\begin{dfn}[$\cc{QIP}(3)$]
	For $\eps: \N \to [0,1]$, the class $\cc{QIP}_\eps (3)$ is the set of languages $L \subseteq \cube*$ for which there exists a three-message quantum verifier $V = (V_n)_n$ (with no output state) satisfying the following conditions:
	\begin{itemize}
		\item \emph{Completeness:} For all $x \in L$, there exists a quantum prover $P$ (called an \emph{honest prover}) such that $\pr{\text{$V(x) \interact P$ accepts}} = 1$.\footnote{The reader may wonder whether the definition of $\QIP(3)$ here is sensitive to the assumption of perfect completeness; it is known that if the verifier uses the universal gate set $\{H, \mathit{CNOT}, T\}$, then we can assume perfect completeness without loss of generality~\cite[Section 4.2]{vidick2016quantum}.}
		\item \emph{Soundness:} For all $x \notin L$ and all quantum provers $P$, it holds that $\pr{\text{$V(x) \interact P$ accepts}} \le \eps(|x|)$.
	\end{itemize}
	Here the probability is over the randomness of the interaction. Define $\cc{QIP}(3) = \bigcap_p \cc{QIP}_{2^{-p}} (3)$ where $p$ ranges over all polynomials.
\end{dfn}

\begin{thm}[Watrous~\cite{watrous03pspace}] \label{thm:watrous}
	$\PSPACE \subseteq \cc{QIP}(3)$.
\end{thm}

We remark that the converse inclusion $\cc{QIP}(3) \subseteq \cc{PSPACE}$ holds as well~\cite{jain2011qip}. It is straightforward to generalize \cref{thm:watrous} from decision problems to functions:

\begin{cor} \label{cor:watrous}
	Let $f: \cube* \to \cube*$ be a $\PSPACE$-computable function such that $|f(x)| \le \poly(|x|)$ for all $x$, and let $\eps$ be a function of the form $\eps(n) = \exp(-\poly(n))$. Then there exists a three-message quantum verifier $V = (V_n)_n$ satisfying the following conditions:
	\begin{itemize}
		\item \emph{Completeness:} For all $x \in \cube*$, there exists a quantum prover $P$ (called an \emph{honest prover}) such that $\pr{\text{$V(x) \interact P$ accepts and outputs $f(x)$}} = 1$.
		\item \emph{Soundness:} For all $x \in \cube*$ and all quantum provers $P$,
		\begin{equation*}
			\pr{\text{$V(x) \interact P$ accepts and outputs a string other than $f(x)$}} \le \eps(|x|).
		\end{equation*}
	\end{itemize}
\end{cor}

\begin{proof}
	The language $L = \{(x,f(x)): x \in \cube*\}$ is clearly in $\PSPACE$, so by \cref{thm:watrous} there exists a $\cc{QIP}_\eps (3)$ verifier $V_L$ for $L$. A verifier $V_f$ for $f$ can be described as follows. First $V_f$ sends the input string $x$ to the prover. Then $V_f$ receives a register $\reg M$ from the prover, measures $\reg M$ in the standard basis to obtain a string $y$, and simulates $V_L$ on input $(x,y)$. (Here the prover is expected to send both $y$ and the first nontrivial message from the simulation of $V_L$ in the same message, so that the total number of nontrivial messages is still three.) If $V_L$ accepts then $V_f$ accepts and outputs $y$, otherwise $V_f$ rejects.
	
	Completeness holds because an honest prover for $V_f$ can send $y = f(x)$ and then simulate an honest prover for $V_L$. Soundness holds because conditioned on any string $y \neq f(x)$ that the verifier measures in $\reg M$, the probability that $V_L$ accepts is at most $\eps(|x|)$ by the soundness of $V_L$.
\end{proof}

\subsection{The classes $\cc{stateQIP}(m)$ and $\cc{stateQIP}$}
\label{sec:stateqip}

\begin{dfn}[{$\stateqip(m)$ and $\stateqip$}]
	\label{def:stateQIP}
	Let $\eps,\delta : \N \to [0,\infty)$ and $m: \N \to \N$ be functions. The class $\stateqip_{\eps,\delta}(m)$ is the set of mixed state sequences $(\rho_n)_n$ (where $\rho_n$ is on $n$ qubits) for which there exists an $m$-message quantum verifier $(V_n)_n$ satisfying the following for all sufficiently large $n$:
	\begin{itemize}
		\item \emph{Completeness:} There exists a quantum prover $P$ (called an \emph{honest prover}) such that $\pr{\text{$V_n \interact P$ accepts}} = 1$.
		\item \emph{Soundness:} For all quantum provers $P$ such that $\pr{\text{$V_n \interact P$ accepts}} \ge \eps(n)$, it holds that $\td(\sigma, \rho_n) \le \delta(n)$ where $\sigma$ denotes the output of $V_n \interact P$ conditioned on accepting.
	\end{itemize}
	Here the probabilities are over the randomness of the interaction.
	
	Finally, define
	\begin{align*}
		&\stateqip(m) = \bigcap_{p,q} \stateqip_{\frac1p, \frac1q} (m),
		&\stateqip = \bigcup_{m^\prime} \stateqip\Paren{m^\prime}
	\end{align*}
	where $p,q,m^\prime$ range over all polynomials.
\end{dfn}

\begin{rmk*}
	Metger and Yuen~\cite{MY23} fixed $p$ to 2 in their definition of $\stateqip$, i.e.\ they considered the class $\stateqip^\prime = \bigcup_m \bigcap_q \stateqip_{\frac12, \frac1q} (m)$. However our definitions are equivalent because
	\begin{equation*}
		\statePSPACE \subseteq \stateqip(6) \subseteq \stateqip \subseteq \stateqip^\prime \subseteq \statePSPACE,
	\end{equation*}
	where the first inclusion is \cref{thm:spq}, the second and third inclusions are trivial, and the fourth inclusion was proved by Metger and Yuen~\cite{MY23}.
\end{rmk*}

\section{Proof that \texorpdfstring{$\statePSPACE \subseteq \cc{stateQIP}(6)$} {statePSPACE in stateQIP}} \label{sec:psq}
In this section we use the background from \cref{sec:scc} to prove \cref{thm:spq}, i.e.\ that $\cc{statePSPACE} \subseteq \cc{stateQIP}(6)$. Let $(\rho_n)_n \in \statePSPACE$ and let $\eps(n), \delta(n) = 1/\poly(n)$; below we prove that $(\rho_n)_n$ is in $\cc{stateQIP}_{\eps, \delta}(6)$ which establishes the theorem.
\subsection{The protocol} \label{sec:prot}

Since $(\rho_n)_n$ is in $\statePSPACE$ there exists a sequence $(\sigma_n)_n \in \statePSPACE_0$ such that $\td\Paren{\rho_n, \sigma_n} \le \delta(n)/2$. By \cref{lem:pur} there exists a sequence of pure states $(\ket{\psi_n})_n \in \sPe$ such that the reduced state on the first $n$ qubits of $\ket{\psi_n}$ equals $\sigma_n$. By \cref{lem:msp} the sequence $(\psi_n)_n$ is $\cc{polyL}$-explicit, so by \cref{lem:exp} the sequence $\Paren{\ket{\psi_n}}_n$ is $\cc{polyL}$-explicit up to global phases. Therefore by \cref{thm:main} there exists a uniform sequence of polynomial-size quantum circuits $(A_n)_n$, making one query to a $\PSPACE$-computable function $f$, such that the reduced state on the initial qubits of $A_n^f \ket\zs$ is within $2^{-n}$ trace distance of $\psi_n$, and furthermore $(A_n)_n$ does not depend on $(\rho_n)_n$. Henceforth we will fix $n$ and write $\rho = \rho_n, \eps = \eps(n)$ and so on for brevity.

Let $m = \poly(n)$ be the number of qubits on which $A$ acts. By the discussion in \cref{sec:prelim}, we can assume without loss of generality that $f$ has a single output bit and that the query in $A^f$ is of the form $D = \sum_{x \in \cube m} (-1)^{f(x)} \kb x$. Write $A^f \ket{0^m} = C D \ket \phi$ where $C$ is the portion of $A$ applied after the query, and $\ket\phi$ is the state constructed by the portion of $A$ applied before the query.

Let $t = \poly(n)$ be a parameter to be chosen later, and for $x_1, \dotsc, x_t \in \cube m$ let $F \Paren{x_1, \dotsc, x_t} = \Paren{f \Paren{x_1}, \dotsc, f \Paren{x_t}}$. Since $f$ is $\PSPACE$-computable, so is $F$. Let $V_F$ be the three-message quantum verifier circuit for $F$ guaranteed to exist by \cref{cor:watrous}, with soundness parameter $2^{-2n}$. As mentioned in \cref{sec:gqp} we can assume without loss of generality that $V_F$ is unitary. We can also assume without loss of generality that $V_F$ preserves the classical state $\ket x$ of its input register, e.g.\ by defining a verifier circuit that makes a copy of $x$ and simulates $V_F$ on the copy.

We name certain registers associated with $V_F$ as follows. Let $\reg A$ be the input register, and write $\reg A = \reg A_1 \dotsb \reg A_t$ where each $\reg A_j$ is an $m$-qubit register. Let $\reg S$ be the output register (which on input $x$, ideally holds $F(x)$), and write $\reg S = \reg S_1 \dotsb \reg S_t$ where each $\reg S_j$ is a one-qubit register. Let $\reg Z$ be the one-qubit register indicating whether to accept or reject, and let $\reg W$ be the register disjoint from $\reg{AZS}$ that holds the rest of the output of $V_F$'s final circuit.

\cref{alg:main} describes a verifier circuit for constructing $\rho$. There are six messages in total, because \cref{line:4} requires four messages (including sending $x$ to the prover) and \cref{line:10,line:11} each require one message.

\begin{algorithm}[t]
	\caption{$\cc{stateQIP}_{\eps, \delta} (6)$ verifier circuit for $\rho$}
	\label{alg:main}
	\begin{algorithmic}[1]
		\For{$j \in [t]$} construct $\ket\phi_{\reg A_j}$.
		\EndFor
		\Ctrl{the classical state $\ket{x}_{\reg A}$,}
			\State Simulate $V_F$ on input $x$. \hlabel{line:4}
		\EndCtrl \hlabel{line:5}
		\If{a standard-basis measurement of $\reg Z$ outputs 0} \textbf{reject} and \textbf{abort}. \hlabel{line:6}
		\EndIf
		\State Sample $k \in [t]$ uniformly at random. \hlabel{line:8}
		\State Apply the Pauli $Z$ matrix in $\reg S_k$. \Comment $Z = \kb0 - \kb1$ \hlabel{line:9}
		\State Send $\reg S \reg W$ to the prover. \hlabel{line:10}
		\State Receive a $t m$-qubit register $\reg M$ from the prover. \hlabel{line:11}
		\Ctrl{the classical state $\ket x_{\reg A}$,}
			\State XOR $x$ into $\reg M$.
		\EndCtrl \hlabel{line:14}
		\For{$j \in [t] \backslash \{k\}$} perform the projective measurement $(\phi, I-\phi)$ in $\reg A_j$. \hlabel{line:15}
			\If{the measurement outcome is $I - \phi$} \textbf{reject} and \textbf{abort}.
			\EndIf
		\EndFor \hlabel{line:18}
		\State Apply $C_{\reg A_k}$. \hlabel{line:19}
		\State \textbf{accept} and \Return the first $n$ qubits of $\reg A_k$.
	\end{algorithmic}
\end{algorithm}

\subsection{Proof of completeness} \label{sec:pc}

We describe an honest prover $P$. On \cref{line:4} $P$ simulates an honest prover $P_F$ for $V_F$. We can assume without loss of generality that if $x$ denotes $V_F$'s input string, then the final state of $P_F$'s workspace includes a copy of $x$ (e.g.\ by having $P_F$ make an extra copy of $x$ at the beginning of its computation). Write $\ket\phi^{\otimes t} = \sum_{x \in \cube{t m}} \alpha_x \ket x$; then we can write the state of the system immediately after \cref{line:4} as
\begin{equation*}
	\sum_{\mathclap{x \in \cube{t m}}} \alpha_x \ket{x}_{\reg A} \ket{F(x)}_{\reg S} \ket1_{\reg Z} \ket{x}_{\reg M} \ket{\theta_x}_{\reg{WQ}}.
\end{equation*}
Here $\reg M$ is a register held by $P$ (which will later be sent to the verifier in \cref{line:11}), the register $\reg Q$ denotes the remainder of $P$'s private workspace, and $\ket{\theta_x}$ is some state.

Let $k$ be the value chosen by the verifier in \cref{line:8}. Given the above state, clearly applying $Z_{\reg S_k}$ has the same effect that applying $D_{\reg A_k}$ would have, so the state of the system after \cref{line:10} is
\begin{equation*}
	D_{\reg A_k} \cdot \sum_{\mathclap{x \in \cube{t m}}} \alpha_x \ket{x}_{\reg A} \ket{F(x)}_{\reg S} \ket{x}_{\reg M} \ket{\theta_x}_{\reg{WQ}}
\end{equation*}
where $\reg A$ is held by the verifier and $\reg{SMWQ}$ is held by $P$.

Next $P$ uncomputes the state $\ket{F(x)}_{\reg S} \ket{\theta_x}_{\reg{WQ}}$ controlled on $\ket{x}_{\reg M}$, and then sends $\reg M$ to the verifier in \cref{line:11}. After \cref{line:14} the verifier holds the state
\begin{equation*}
	D_{\reg A_k} \cdot \sum_{\mathclap{x \in \cube{t m}}} \alpha_x \ket{x}_{\reg A}
	= D_{\reg A_k} \cdot \bigotimes_{j \in [t]} \ket\phi_{\reg A_j},
\end{equation*}
which clearly passes the subsequent measurements with probability 1.

\subsection{Proof of soundness}

It will be convenient to refer to the output register in a manner independent of the random variable $k$ from \cref{line:8}. To this end, let $\reg O$ be an $m$-qubit register, and imagine that the verifier's final action is to apply the channel $\Phi_k$ that acts as the identity on the system except that $\Phi_k$ renames $\reg A_k$ as $\reg O$. Fix a prover such that the verifier accepts with probability $\eps^\prime \ge \eps$. Let $\tau$ denote the state of the system at the end of the protocol, conditioned on accepting, and let $\tau^O$ denote the reduced state of $\tau$ on $\reg O$. Then $\ptr{>n}{\tau^O}$ is the output state conditioned on accepting.

Let $n^\prime$ be the number of qubits comprising $\ket\psi$. By the triangle inequality, \cref{eq:tdc,eq:fvdg}, and various definitions from \cref{sec:prot}, it holds that
\begin{equation} \label{eq:sound1}
	\td\Paren{\trg n {\tau^O}, \rho}
	\le \td\Paren{\trg n {\tau^O}, \sigma} + \td(\sigma, \rho)
	\le \td\Paren{\trg n {\tau^O}, \trg n \psi} + \delta/2
\end{equation}
and that
\begin{align}
	\td\Paren{\trg n {\tau^O}, \trg n \psi}
	&\le \td\Paren{\trg {n^\prime} {\tau^O}, \psi} \nonumber \\
	&\le \td\Paren{\ptr{>n^\prime}{\tau^O}, \ptr {>n^\prime} {C D \phi D \adj C}} + \td\Paren{\ptr {>n^\prime} {C D \phi D \adj C}, \psi}\nonumber  \\
	&\le \td\Paren{\tau^O, C D \phi D \adj C} + 2^{-n}
	\le \sqrt{\tr{\tau \cdot (I - C D \phi D \adj C)_{\reg O}}} + 2^{-n}. \label{eq:sound2}
\end{align}

Let $\ket\varphi$ denote the state of the system after \cref{line:5}, and let $U$ be the unitary jointly applied by the verifier and prover from \cref{line:10} to \cref{line:14}. Then
\begin{equation*}
	\eps^\prime \tau = \frac1t \sum_{k=1}^t \Phi_k \Paren{\theta_k}
	\qquad \text{for} \qquad
	\ket{\theta_k} = \bigotimes_{j \neq k} \bra\phi_{\reg A_j} \cdot C_{\reg A_k} U Z_{\reg S_k} \bra1_{\reg Z} \ket\varphi,
\end{equation*}
where $\ket{\theta_k}$ is (in general) subnormalized and $\theta_k = \kb{\theta_k}$. Let
\begin{equation*}
	Q = \sum_{\mathclap{x \in \cube{t m}}} x_{\reg A} \otimes F(x)_{\reg S},
\end{equation*}
and similarly define a matrix $\tilde\tau$ as follows:
\begin{equation*}
	\eps^\prime \tilde\tau = \frac1t \sum_{k=1}^t \Phi_k \Paren{\tilde \theta_k}
	\qquad \text{for} \qquad
	\ket{\tilde\theta_k} = \bigotimes_{j \neq k} \bra\phi_{\reg A_j} \cdot C_{\reg A_k} U Z_{\reg S_k} Q \bra1_{\reg Z} \ket\varphi.
\end{equation*}

We now argue that $\tilde\tau$ is a close approximation of $\tau$, using the soundness property of $V_F$. For $k \in [t]$ it holds that $\Norm{\ket{\tilde\theta_k} - \ket{\theta_k}}^2 \le \Norm{(I-Q) \bra1_{\reg Z} \ket\varphi}^2$. This bound equals the probability that if the register $\reg{ASZ}$ of $\ket\varphi$ is measured in the standard basis, then the measurement outcome is of the form $\ket{x}_{\reg A} \ket{y}_{\reg S} \ket 1_{\reg Z}$ where $y \neq F(x)$. Conditioning on $x$ and applying the soundness of $V_F$ shows that this event has probability at most $2^{-2n}$, so $\Norm{\ket{\tilde\theta_k} - \ket{\theta_k}} \le 2^{-n}$. Therefore by the triangle inequality,
\begin{align*}
	\eps^\prime \Norm{\tilde\tau - \tau}_1
	&\le \frac1t \sum_{k=1}^t \Norm{\kb{\tilde\theta_k} - \kb{\theta_k}}_1 \\
	&\le \frac1t \sum_{k=1}^t \Paren{\Norm{\Paren{\ket{\tilde\theta_k} - \ket{\theta_k}} \bra{\tilde\theta_k}}_1 + \Norm{\ket{\theta_k} \Paren{\bra{\tilde\theta_k} - \bra{\theta_k}}}_1} \\
	&\le \frac1t \sum_{k=1}^t 2^{-n} \Paren{\Norm{\ket{\tilde\theta_k}} + \Norm{\ket{\theta_k}}}
	\le 2 \cdot 2^{-n}.
\end{align*}
Since $\eps^\prime \ge \eps \ge 1/\poly(n)$ it follows that $\Norm{\tilde\tau - \tau}_1 \le \exp(-\Omega(n))$.

Let $P = (I - CD\phi D \adj C)_{\reg O}$. Since $P$ is an orthogonal projection,
\begin{equation}
	\tr{\tau P}
	\le \tr{\tilde\tau P} + \frac {\Norm{\tilde\tau - \tau}_1} 2
	\le \tr{\tilde\tau P} + \exp(-\Omega(n)). \label{eq:sound3}
\end{equation}
By reasoning similar to that in \cref{sec:pc}, it holds that $U Z_{\reg S_k} Q = U D_{\reg A_k} Q = D_{\reg A_k} U Q$, so defining the subnormalized vector $\ket{\varphi^\prime} = U Q \bra1_{\reg Z} \ket\varphi$ it holds that
\begin{equation*}
	\ket{\tilde\theta_k} = \bigotimes_{j \neq k} \bra\phi_{\reg A_j} \cdot (CD)_{\reg A_k} \ket{\varphi^\prime}.
\end{equation*}
Therefore since trace is linear,
\begin{align*}
	\eps^\prime \tr{\tilde\tau P}
	&= \frac1t \sum_{k=1}^t \tr{\Phi_k \Paren{\tilde\theta_k} P}
	= \frac1t \sum_{k=1}^t \tr{\tilde\theta_k \cdot (I - CD\phi D \adj C)_{\reg A_k}} \\
	&= \frac1t \tr{\varphi^\prime \cdot \sum_{k=1}^t \bigotimes_{j \neq k} \phi_{\reg A_j} \otimes \Paren{I - \phi}_{\reg A_k}}
	\le \frac1t \tr{\varphi^\prime}
	\le \frac1t,
\end{align*}
where we used that $\sum_{k=1}^t \bigotimes_{j \neq k} \phi_{\reg A_j} \otimes \Paren{I - \phi}_{\reg A_k}$ is an orthogonal projection. Since $\eps^\prime \ge \eps$ it follows that
\begin{equation}
	\tr{\tilde\tau P} \le 1/(\eps t). \label{eq:sound4}
\end{equation}

Choose $t = \ceil*{16/\Paren{\eps \delta^2}} \le \poly(n)$. Then for all sufficiently large $n$, it follows from \cref{eq:sound1,eq:sound2,eq:sound3,eq:sound4} that
\begin{equation*}
	\td\Paren{\ptr{>n}{\tau^O}, \rho}
	\le \sqrt{\frac1{\eps t} + \exp(-\Omega(n))} + 2^{-n} + \delta/2
	\le \frac2{\sqrt{\eps t}} + \frac\delta2
	\le \delta.
\end{equation*}

\section{Barrier to \qaczf lower bounds for constructing explicit states} \label{sec:qacfz}

Call a state sequence $(\ket{\psi_n})_n$ \emph{explicit} if $\ket{\psi_n}$ is an $n$-qubit state whose description can be computed in time $\exp(\poly(n))$ as a function of $n$. For example, every pure state sequence in $\sPe$ is explicit up to global phases, by \cref{lem:exp,lem:msp} and the fact that $\cc{PSPACE} \subseteq \cc{EXP}$. We say that a language is in \qaczf if it can be decided with bounded error by a nonuniform sequence of polynomial-size \qaczf circuits. The following is one way to more formally state \cref{obs:bar}:

\begin{thm}
	Assume there exists an explicit state sequence $\Paren{\ket{\psi_n}}_n$ and function $\eps(n) = \exp(-\poly(n))$ such that for all sequences $(C_n)_n$ of polynomial-size \qaczf circuits, it holds that $\Norm{C_n \ket\zs - \ket{\psi_n} \ket\zs} \ge \eps(n)$. Then $\cc{EXP} \nsubseteq \cc{QAC_f^0}$.
\end{thm}

\begin{proof}
	We prove the contrapositive statement: if $\cc{EXP} \subseteq \cc{QAC_f^0}$ then for all functions $\eps(n) = \exp(-\poly(n))$, every explicit state sequence $(\ket{\psi_n})_n$ can be constructed to within error $\eps$ in \qaczf. Let $C_n^{f_n}$ be the circuit-oracle combination for constructing $\ket{\psi_n}$ from \cref{thm:fqa}. We argue that $(f_n)_n$ is in $\cc{EXP}$: given $n$, first compute the description of $\ket{\psi_n}$ (which takes exponential time since $(\ket{\psi_n})_n$ is explicit) and then run the assumed algorithm for $f_n$ from \cref{thm:fqa} (which takes polynomial space and therefore exponential time). By the assumption that $\cc{EXP} \subseteq \cc{QAC_f^0}$ it follows that $(f_n)_n \in \cc{QAC_f^0}$, and therefore $\Paren{C_n^{f_n}}_n$ can be implemented in \qaczf.
\end{proof}

\section{Approximately constructing arbitrary states} \label{sec:states}
\subsection{Upper bound}

\rubs*

\begin{proof}
	Let $\mc G$ be any universal gate set that includes the Toffoli and NOT gates. By \cref{thm:tqa} and the Solovay--Kitaev theorem~\cite{BG21,DN06} there exists a $\poly(n)$-size circuit $A$ over $\mc G$, making ten queries to a Boolean function $f$, such that $\Norm{A^f \ket\zs - \ket\psi \ket\zs} \le \eps$. Inspection of the proof of \cref{thm:tqa} reveals that $f$ has $n + \log \log(1/\eps) + O(1)$ input bits, and that only the first output bit of $f$ depends on the input to $f$. For all $m$ every function from $m$ bits to 1 bit can be computed by an $O(2^m/m)$-size Boolean circuit~\cite{Juk12,Lup58}, so $f$ can be computed by an $O(2^n \log(1/\eps) / n)$-size Boolean circuit, where the output bits not depending on the input are hard-coded into the circuit. Since Boolean circuits can be cleanly simulated by quantum circuits consisting only of Toffoli and NOT gates with a constant-factor blowup in size, it follows that $f$ can be computed by an $O(2^n \log(1/\eps) / n)$-size circuit over $\mc G$. Combining this circuit with $A$ yields the desired result.
\end{proof}

\subsection{Lower bound}

\rlbs*

\begin{proof}	
	Let $S_n(r) = \{x \in \R^{n+1}: \norm x = r\}$ and $S_n = S_n(1)$. The set of $n$-qubit pure states can be identified with $S_{2^{n+1}-1}$, because an $n$-qubit pure state is described by $2^n$ complex amplitudes, each of which has a real part and an imaginary part, and these $2^{n+1}$ real numbers form a unit vector. Let $\mu_n$ denote $n$-dimensional volume; then $\mu_n (S_n)$ obeys the recurrence
	\begin{align*}
		&\mu_0(S_0) = 2,&
		&\mu_1(S_1) = 2\pi,&
		&\mu_{n+1} (S_{n+1}) = 2\pi \mu_{n-1}(S_{n-1}) / n \quad \text{for $n \ge 1$}
	\end{align*}
	and $\mu_n(S_n(r)) = r^n \mu_n(S_n)$~\cite{Wik23}. We will write $\mu = \mu_n$ when $n$ is clear from the context.
	
	For an $n$-qubit mixed state $\rho$ and $\eps \ge 0$, let $N_\eps(\rho)$ denote the set of pure states $\ket\psi$ such that $\td(\rho, \psi) \le \eps$. If $\rho$ itself is rank-1, say $\rho = \kb\rho$, then for all pure states $\ket\psi$ it is well known that $\td(\rho, \psi) = \sqrt{1 - |\ip \rho \psi|^2}$, and so $\ket\psi$ is in $N_\eps(\rho)$ if and only if $|\ip \rho \psi|^2 \ge 1-\eps^2$. Therefore
	\begin{equation*}
		\mu\Paren{N_\eps(\rho)} = \int_{\theta=0}^{\arcsin\eps} \mu\Paren{S_1(\cos \theta)} \mu\Paren{S_{2^{n+1} - 3}(\sin \theta)} d\theta,
	\end{equation*}
	because $\ip\rho\psi$ is described by two real numbers whose squares sum to a value $\cos^2 \theta$ between $1$ and $1-\eps^2$, and the rest of $\ket\psi$ is described by $2^{n+1} - 2$ real numbers whose squares sum to $\sin^2 \theta$. It follows that for $m = 2^{n+1}$,
	\begin{align*}
		\mu\Paren{N_\eps(\rho)}
		&= \int_{\theta=0}^{\arcsin\eps} \cos \theta \sin^{m-3} \theta d\theta \cdot \mu(S_1) \mu(S_{m-3})
		= \int_{u=0}^\eps u^{m-3} du \cdot \mu(S_1) \mu(S_{m-3}) \\
		&= \eps^{m-2} \mu(S_1) \mu(S_{m-3}) / (m-2)
		= \eps^{m-2} \mu(S_{m-1}).
	\end{align*}
	
	More generally, consider an $n$-qubit mixed state $\rho$ of arbitrary rank. If $N_\eps(\rho)$ is nonempty then there exists a state $\ket\psi \in N_\eps(\rho)$, so for all $\ket\phi \in N_\eps(\rho)$, by the triangle inequality $\td(\psi, \phi) \le \td(\psi, \rho) + \td(\rho, \phi) \le 2\eps$. In other words $N_\eps(\rho) \subseteq N_{2\eps} (\psi)$. It follows from the case proved above that
	\begin{equation*}
		\mu\Paren{N_\eps(\rho)}
		\le \mu\Paren{N_{2\eps}(\psi)}
		\le (2\eps)^{m-2} \mu(S_{m-1})
		\le \eps^{(m-2)/2} \mu(S_{m-1}),
	\end{equation*}
	where the last inequality holds because $\eps \le 1/4$.

	For $s \in \N$ let $\mc C_s$ denote the set of circuits over $\mc G$ consisting of $s$ gates. Circuits in $\mc C_s$ act on $O(s)$ qubits without loss of generality, and there are $\poly(s)$ ways to choose a gate from $\mc G$ and the qubits that it acts on out of $O(s)$ total qubits, so $\Mag{\mc C_s} \le \poly(s)^s \le 2^{O(s \log s)}$. In particular, if $s \le o(2^n \log(1/\eps) / n)$ then $\log s \le O(n) + \log \log(1/\eps) \le O(n)$ and so $2^{O(s \log s)} \le (1/\eps)^{o(2^n)}$; therefore
	\begin{align*}
		&\mu\Paren{\bigcup_{C \in \mc C_s} N_\eps \Paren{\trg n {C \kb\zs \adj C}}}
		\le \sum_{C \in \mc C_s} \mu\Paren{N_\eps \Paren{\trg n {C \kb\zs \adj C}}} \\
		&\quad\le \sum_{C \in \mc C_s} \eps^{(m-2)/2} \mu(S_{m-1})
		\le \eps^{(m-2)/2 - o(m)} \mu(S_{m-1})
		\le o\Paren{\mu(S_{m-1})}. \qedhere
	\end{align*}
\end{proof}

\appendix
\crefalias{section}{appendix}

\section{Proof of \texorpdfstring{\cref{lem:a5}}{Lemma 3.2}} \label{app:a5}

Recall that in \cref{sec:halg} we defined $\alpha = 0.35$ and
\begin{equation*}
	\ket{p_{\eta, C}} = C \cdot 2^{-n/2} \sum_{\mathclap{x \in \cube n}} \sr(\bra\eta C \ket x) \ket x
\end{equation*}
for a Clifford unitary $C$ and vector $\ket\eta \in \Paren{\C^2}^{\otimes n}$. We establish the following fact:

\raf*

Eq.~(A.22) of Irani et al.~\cite{INN+22}---where their $\ket\tau$ equals our $\ket\eta$, their $\gamma$ can be set to $0.24999$, and their $d$ equals $2^n$---implies that
\begin{equation*}
	\PR{\Norm{\mrm{Re} \Paren{C \ket\eta}}_1 \ge \sqrt{0.24999 \cdot 2^n}} > 0
\end{equation*}
for a random Clifford unitary $C$. Therefore there exists a fixed Clifford unitary $C$ such that $2^{-n/2} \Norm{\mrm{Re} \Paren{\adj C \ket\eta}}_1 \ge 0.4999$. Finally it follows from the definition of $\ket{p_{\eta, C}}$ that
\begin{equation*}
	\mrm{Re}(\ip \eta {p_{\eta, C}}) = 2^{-n/2} \sum_x \Mag{\mrm{Re}(\bra\eta C \ket x)} = 2^{-n/2} \Norm{\mrm{Re}(\adj C \ket\eta)}_1,
\end{equation*}
implying that \cref{lem:a5} holds with $\alpha = 0.4999$.

We instead define $\alpha = 0.35$ because we believe that there is a typo in Irani et al.~\cite{INN+22}, and that the right side of their Eq.~(A.22) should be $1/2 - 4\gamma$ instead of $1/2 - 2\gamma$. So in the above analysis we should actually set $\gamma$ to be slightly less than $1/8$, and so the value of $\alpha$ should be slightly less than $\sqrt{1/8} \approx 0.354$. The exact value of $\alpha$ is not important for our main results however.

Our disagreement with the argument in Irani et al.~\cite{INN+22} is as follows. We will use their notation; in particular they assign a different meaning to the variable $\alpha$ than we have done. First---and this part is actually an understatement by Irani et al., not an error---in Eq.~(A.13) the expression $\sqrt{\Paren{2^n+1} / \Paren{2\alpha}}$ can trivially be replaced by $\sqrt{\Paren{2^n+1} / (4 \alpha)}$, and so Eq.~(A.15) can be replaced by ``$\ge 1 - 1/(2\alpha)$". Applying this strengthening of Eq.~(A.15) with $\alpha = 1/(16\gamma)$ implies that $\PR{\Norm{\ket\psi}_1 \ge 2\sqrt{\gamma 2^n}} \ge 1 - 8\gamma$, where $\ket\psi$ is as defined in Lemma A.5 of Irani et al.

Write $\ket\psi = \ket{a} + i \ket{b}$ where $\ket a, \ket b \in \R^{2^n}$. Then
\begin{equation*}
	\PR{\norm{\ket a}_1 \ge \sqrt{\gamma 2^n}}
	\ge \PR{\norm{\ket a}_1 \ge \sqrt{\gamma 2^n}  \, \middle| \, \Norm{\ket\psi}_1 \ge 2\sqrt{\gamma 2^n}}
	\PR{\Norm{\ket\psi}_1 \ge 2\sqrt{\gamma 2^n}}.
\end{equation*}
By the triangle inequality $\Norm{\ket\psi}_1 \le \Norm{\ket a}_1 + \Norm{\ket b}_1$, so conditioned on $\Norm{\ket\psi}_1 \ge 2\sqrt{\gamma 2^n}$ either $\Norm{\ket a}_1 \ge \sqrt{\gamma 2^n}$ or $\Norm{\ket b}_1 \ge \sqrt{\gamma 2^n}$ (or both). Furthermore $\Norm{\ket a}_1$ and $\Norm{\ket b}_1$ are identically distributed conditioned on any value of $\Norm{\ket \psi}_1$, because applying a global phase of $i$ to $\ket\psi$ has the effect of swapping $\Norm{\ket a}_1$ and $\Norm{\ket b}_1$ without changing $\Norm{\ket \psi}_1$. Therefore
\begin{equation*}
	\PR{\norm{\ket a}_1 \ge \sqrt{\gamma 2^n}  \, \middle| \, \Norm{\ket\psi}_1 \ge 2\sqrt{\gamma 2^n}} \ge 1/2
\end{equation*}
and so $\PR{\norm{\ket a}_1 \ge \sqrt{\gamma 2^n}} \ge \frac12\Paren{1 - 8\gamma} = 1/2 - 4\gamma$.

\section{State synthesis using perfect linear hash functions} \label{app:var}

Below we argue that in the proof of (a statement similar to) \cref{lem:help}, instead of using states of the form $C \cdot 2^{-n/2} \sum_{x \in \cube n} \pm \ket x$ where $C$ is a Clifford unitary, we could alternatively use what we call ``hash states":

\begin{dfn}[Hash states]
	A hash state is an $n$-qubit state $\ket\phi$ such that there exists a set $S \subseteq \cube n$, with $|S| = 2^k$ a power of 2, such that $\ket\phi = |S|^{-1/2} \sum_{x \in S} \sigma_x \ket x$ where $\sigma_x \in \{1,-1\}$, and furthermore there exists a linear transformation $A: \ft^n \to \ft^k$ that is one-to-one on $S$. In particular if $k=0$ then $A$ exists vacuously.
\end{dfn}

\begin{rmk*}
	The resulting variant of \cref{lem:help} would have a lower \qaczf circuit depth, but a measurement of the first $t$ qubits would output $0^t$ with probability $\Theta(1/n)$ instead of $\Theta(1)$. This is not a problem for our proofs of \cref{thm:mi,thm:spq,obs:bar}, but would be in our proof of \cref{thm:ubs}.
\end{rmk*}

A hash state $\ket\phi$ can be constructed with one query as follows. First prepare $\ket{+^k}$ in a register $\reg R$. Then controlled on the state $\ket y_{\reg R}$ where $y \in \cube k$, query the unique string $x \in S$ such that $Ax = y$, while simultaneously making a query to apply a phase of $\sigma_x$. Finally use $A$ to uncompute $y$ controlled on $x$, using that parity is in \qaczf~\cite{Gre+02}.

More generally, for $0 \le j < T$ let $\ket{\phi_j}$ be a hash state and let $A_j \in \ft^{k_j \times n}$ be the linear transformation associated with $\ket{\phi_j}$. To construct $\ket{\phi_j}$ controlled on $j$, first construct $\ket{+^n}_{\reg R}$, and then proceed as above controlled on $j$. Here the oracle ignores the last $n - k_j$ qubits of $\reg R$, and also outputs descriptions of $A_0, \dotsc,A_{T-1}$. Finally uncompute $\ket{+^{n-k_j}}$ in the last $n-k_j$ qubits of $\reg R$, controlled on $k_j$ (which is implicit in the description of $A_j$).

All that remains is to write an arbitrary $n$-qubit state $\ket\psi$ as a linear combination of hash states, in a manner suitable to an LCU application like that in \cref{sec:oqa}. (It will be apparent from our proof that the queries can be computed in $\poly(n)$ space given the description of $\ket\psi$, by reasoning similar to that in \cref{sec:qps}.) By writing $\ket\psi = \ket{\psi_R} + i \ket{\psi_I}$ where $\ket{\psi_R}$ and $\ket{\psi_I}$ are real-valued vectors, it suffices to write a \emph{real-valued} vector with norm at most 1 as such a linear combination of hash states. To do this we will need the following lemma, which is proved using the probabilistic method:

\begin{lem} \label{lem:ft}
	Let $n>0$. For all $S \subseteq \ft^n$ with $|S|=2^k$ a power of $2$, there exists a matrix $A \in \ft^{k \times n}$ satisfying $|\{Ax: x \in S\}| > \frac12 \cdot 2^k$.
\end{lem}

We remark that Alon, Dietzfelbinger, Miltersen, Petrank and Tardos~\cite{ADM+99} also investigated the properties of random linear hash functions from $S \subseteq \ft^n$ to $\ft^k$. However, they did not bound the number of nonempty buckets when $|S| = 2^k$.

\begin{proof}[Proof of \cref{lem:ft}]
	Let $A \in \ft^{k \times n}$ be uniform random conditioned on having rank $k$. The kernel of $A$ has dimension $n-k$ and therefore contains $2^{n-k}$ elements, one of which is the all-zeros vector. Therefore any fixed nonzero vector is in $\ker(A)$ with probability $p \coloneqq \frac{2^{n-k} - 1}{2^n - 1}$.\footnote{One way to see this is as follows. Let $x,y \in \ft^n$ be nonzero vectors, and let $B \in \ft^{n \times n}$ be an invertible matrix such that $Bx = y$. Then $AB$ is distributed identically to $A$, so $\pr{Ax = 0} = \pr{ABx = 0} = \pr{Ay = 0}$.} We say that distinct strings $x,y \in S$ \emph{collide} if $Ax = Ay$. Since any distinct $x, y \in S$ collide with probability $\pr{A(x+y) = 0} = p$, the expected number of collisions is
	\begin{equation*}
		\binom{|S|}2 \cdot p
		= \frac{2^k (2^k - 1)} 2 \cdot \frac{2^{n-k} - 1}{2^n - 1}
		= \frac{2^k - 1}2 \cdot \frac{2^n - 2^k}{2^n-1}
		< \frac{2^k}2.
	\end{equation*}
	Therefore there exists a fixed matrix $A$ with less than $2^k/2$ collisions.
	
	Let $T = \{Ax: x \in S\}, t = |T|$ and for $y \in T$ let $S_y = \{x \in S: Ax = y\}$. The sets $S_y$ form a partition of $S$, so by Jensen's inequality the number of collisions is\footnote{Define $\binom r 2 = r(r-1)/2$ even for non-integer values of $r$.}
	\begin{equation*}
		\sum_{y \in T} \binom{|S_y|}2
		\ge t \cdot \binom{\sum_{y \in T} |S_y|/t}2
		= t \cdot \binom{2^k/t}2
		= \frac{2^k}2 \cdot \Paren{\frac{2^k}{t} - 1}.
	\end{equation*}
	Since $2^k/2$ is greater than the number of collisions which is at least $2^k/2 \cdot (2^k/t - 1)$, it follows that $t > 2^k/2$.
\end{proof}

Using \cref{lem:ft} we prove the following:

\begin{lem} \label{lem:ar}
	For all $n$-qubit states $\ket\psi$ with real (standard-basis) amplitudes, there exists a hash state $\ket\phi$ such that $\ip \phi \psi \ge \Omega(1/\sqrt n)$.
\end{lem}

\begin{proof}[Proof of \cref{lem:ar}]
	Write $\ket\psi = \sum_{x \in \cube n} \alpha_x \ket x$. By a limiting argument we can assume without loss of generality that the $|\alpha_x|$ are all distinct. Let $0 \le k \le n$ be a parameter to be chosen later, and let $S$ be the set of the $2^k$ largest elements of $\cube n$ according to the total order defined by $x > y$ when $|\alpha_x| > |\alpha_y|$. By \cref{lem:ft} there exists a matrix $A \in \ft^{k \times n}$ such that $|\{Ax: x \in S\}| > \frac12 \cdot 2^k$.
	
	Define a function $f: \cube k \to \cube n$ as follows: for all $y \in \cube k$, if there exists $x \in S$ such that $Ax = y$ then let $f(y)$ be the lexicographically first $x \in S$ such that $Ax = y$, and otherwise let $f(y)$ be the lexicographically first $x \in \cube n$ such that $Ax = y$. (To see that such an $x$ exists in the latter case, note that $A$ has rank $k$ because the image of $A$ has cardinality greater than $2^{k-1}$.) Let $\ket\phi = 2^{-k/2} \sum_{x \in \im f} \sgn(\alpha_x) \ket{x}$ where $\im f$ denotes the image of $f$. Clearly $\ket\phi$ is a hash state, and
	\begin{align*}
		\ip\phi\psi
		&= 2^{-k/2} \sum_{\mathclap{x \in \im f}} |\alpha_x|
		\ge 2^{-k/2} \sum_{\mathclap{x \in \im f \cap S}} |\alpha_x|
		\ge 2^{-k/2} \cdot |\im f \cap S| \cdot \min_{x \in S} |\alpha_x| \\
		&= 2^{-k/2} \cdot |\{Ax: x \in S\}| \cdot \min_{x \in S} |\alpha_x|
		\ge \frac12 \cdot 2^{k/2} \cdot \min_{x \in S} |\alpha_x|.
	\end{align*}
	
	For $j \in [2^n]$ let $\beta_j$ be the $j$'th largest element of the set $\{|\alpha_x|: x \in \cube n\}$, and let $\mu = \max_{j \in [2^n]} \beta_j \sqrt j$. Then
	\begin{equation*}
		1 = \sum_{\mathclap{x \in \cube n}} \alpha_x^2
		= \sum_{j=1}^{2^n} \beta_j^2
		\le \sum_{j=1}^{2^n} (\mu/\sqrt j)^2
		= \mu^2 \sum_{j=1}^{2^n} 1/j
		\le O(\mu^2 n),
	\end{equation*}
	so $\mu \ge \Omega(1/\sqrt n)$. Let $j \in [2^n]$ be such that $\mu = \beta_j \sqrt j$, and choose $k$ such that $2^k \le j < 2^{k+1}$. Then,
	\begin{equation*}
		\ip\phi\psi
		\ge \frac12 \cdot \sqrt{2^k} \cdot \beta_{2^k}
		\ge \frac12 \cdot \sqrt{j/2} \cdot \beta_j
		= \frac1{2 \sqrt 2} \cdot \mu
		\ge \Omega(1/\sqrt n). \qedhere
	\end{equation*}
\end{proof}

\printbibliography[heading=bibintoc]

\end{document}